\documentclass[notitlepage]{article}%
\usepackage{harvard}
\usepackage{amsmath, bm}
\usepackage{graphicx}
\usepackage{amsfonts}
\usepackage{amssymb}
\usepackage{relsize}
\usepackage{scalefnt}
\usepackage{amsmath}%
\providecommand{\U}[1]{\protect\rule{.1in}{.1in}}
\harvardparenthesis{none}
\newtheorem{theorem}{Theorem}
\newtheorem{lemma}{Lemma}

\advance \topmargin by -\headheight
\advance \topmargin by -\headsep
\evensidemargin \oddsidemargin
\begin{document}

\title{Testing Hypotheses About Ratios of Linear Trend Slopes in Systems of Equations
with a Focus on Tests of Equal Trend Ratios}
\author{Timothy J. Vogelsang\thanks{Tim Vogelsang, Department of Economics, 486 W.
Circle Drive, 110 Marshall-Adams Hall, Michigan State University East Lansing,
MI 48824-1038. Phone: 517-353-4582, Fax: 517--432-1068, email: tjv@msu.edu}}
\date{February 26, 2026}
\maketitle

\begin{abstract}
This paper develops inference methods for ratios of deterministic trend slopes
in systems of pairs of time series. Hypotheses based on linear cross-equation
restrictions are considered with particular interest in tests that trend
ratios are equal across pairs of trending series. Tests of equal ratios can be
used for the empirical assessment of climate models through comparisons of
trend ratios (amplification ratios) of model generated temperature series and
observed temperature series. The analysis in this paper builds on the
estimation and inference methods developed by \nocite{vogelsang-nawaz-JTSA}%
Vogelsang and Nawaz (2017, \textit{Journal of Time Series Analysis}) for a
single pair of trending time series. Because estimators of ratios can have
poor finite sample properties when the trend slope are small relative to
variation around the trends, tests of equal trend ratios are restated in terms
of products of trend slopes leading to inference that is less affected by
small trend slopes. Asymptotic theory is developed that can be used to
generate critical values. For tests of equal trend ratios, finite sample
performance is assessed using simulations. Practical advice is provided for
empirical practitioners. An empirical application compares amplification
ratios (trend ratios) across a set of five groups of observed global
temperature series.

\end{abstract}

\bigskip

\setcounter{page}{0} \thispagestyle{empty} \baselineskip=16.0pt \pagestyle
{plain}

\section{Introduction}

This paper analyzes estimation and inference methods for parameters that
represent ratios of pairs of linear trend slopes in a system of linear
trending time series with covariance stationary fluctuations around the trend.
The proposed methods extend the results of \cite{vogelsang-nawaz-JTSA} where
the focus was on a single trend ratio parameter, and it was shown that
instrumental variables (IV) estimation using time as an instrument was
preferred to ordinary least squares (OLS) of the relevant estimating equation.
Here, there is a vector of trend ratio parameters estimated by IV. Inference
focuses on hypotheses that represent linear restrictions across trend ratio
parameters. Much of the intuition in \cite{vogelsang-nawaz-JTSA} for the
single pair case extends to the multi-pair case especially with respect to the
importance of the magnitude of trend slopes relative to noise
(variation/fluctuations) around the trends for estimation and inference.
Throughout the paper, magnitudes of trends slopes are always interpreted
relative to the magnitude of the noise.

Detailed attention is given to the special case where an empirical
practitioner wants to test the equality of trend ratio parameters between two
pairs of series. This special case is directly relevant for the study of
amplification ratios in the empirical climate literature;
\nocite{santer2005amplification}Santer \textit{et al} (2005),
\nocite{thorne2007tropical}Thorne \textit{et al} (2007),
\nocite{klotzbach2009alternative}\nocite{klotzbach2010correction}Klotzbach
\textit{et al} (2009, 2010), \nocite{christy2010observational}Christy
\textit{et al} (2010), \cite{po2012discrepancies}, \cite{vogelsang-nawaz-JTSA}%
, \cite{VMCS_2026} (hereafter VMCS26), and references in those papers.

The amplification ratio is the ratio of a temperature trend in the troposphere
of the earth relative to a temperature trend at the surface of the earth. A
key assessment of theoretical climate models is determining whether estimated
amplification ratios of observed temperature series are aligned with estimated
amplification ratios of model generated temperatures. The null hypothesis of
interest is equality of trend ratios between two pairs of temperature series:
one pair for observed temperatures and a second pair for model generated
temperatures. Whereas the previous literature treats the amplification ratio
computed for model generated temperatures as fixed (ignoring that it is
estimated), the methods developed in this paper treat both the observed and
model amplification ratios as estimators. Inference takes into account the
joint sampling distributions of both estimated amplification ratios. VMCS26
use these methods to assess the alignment of amplification ratios in CMIP6
model generated temperature series with amplification ratios in observed
temperature series. VMCS26 find systematic misalignments especially for
amplification ratios in the lower troposphere in which case the models exhibit
substantially higher amplification (more warming relative to the surface) than
in observed temperatures. In the present paper amplification ratios are
compared across five sets of observed temperatures.

The remainder of the paper is organized as follows: Section 2 lays out the
system of estimating equations used to estimate trend ratios of pairs of
linear trending time series. The asymptotic properties of the IV estimator of
the trend slope ratios is provided. As shown by \cite{vogelsang-nawaz-JTSA},
IV estimation is used rather than OLS because of systematic correlation
between the regressors and regression errors. Section 3 provides a framework
for testing linear restrictions of trend ratios across equations in the
system. Test statistics are configured to be robust to serial correlation in
the fluctuations of the time series around their trends as well as correlation
across time series. Tests of equal trend ratios between pairs is obtained as a
special case. When the trend slopes are small relative to the variation of the
time series around their trend, the IV estimators, and tests built on them,
can behave very differently than when trend slopes are large. To help offset
this sensitivity to the magnitude of the trend slopes, Section 4 explores an
alternative approach to inference for tests of equal trends that is labeled
the "product approach". The idea is related to the linear-in-trend slopes
inference method (\nocite{fieller1954}Fieller 1954) used by
\cite{vogelsang-nawaz-JTSA} for a single trend ratio. Unlike the linear in
trend slopes approach, the product approach is not fully robust to very small,
or even zero, trend slopes. However, as the finite sample simulations show in
Section 5, the product approach for testing equal trend slopes is often less
sensitive to very small trend slopes relative to tests based on the IV
estimators. The simulations suggest that while the product approach can be
more robust to very small trend slopes, IV based tests tend to be less
sensitive to strong autocorrelation (tendency to over-reject under the null
hypothesis is less). Section 6 provides some practical recommendations for
empirical researchers. Section 7 uses the proposed tests of equal trend ratios
to compare amplification ratios across five sets of observed temperatures for
the tropics of the earth. Section 8 concludes, and proofs are provided in an appendix.

\section{The Model and Estimation}

\subsection{Statistical Model and Assumptions}

The setup consists of $i=1,2,\ldots,n$ \textit{pairs} of univariate linear
trending time series, $y_{1t}^{(i)}$ and $y_{2t}^{(i)}$, given by%
\begin{equation}
y_{1t}^{(i)}=\mu_{1}^{(i)}+\beta_{1}^{(i)}t+u_{1t}^{(i)}, \label{1.1}%
\end{equation}%
\begin{equation}
y_{2t}^{(i)}=\mu_{2}^{(i)}+\beta_{2}^{(i)}t+u_{2t}^{(i)}, \label{1.2}%
\end{equation}
where $u_{1t}^{(i)}$ and $u_{2t}^{(i)}$ are mean zero covariance stationary
processes and $t=1,2,\ldots,T$. The parameters of interest are the $n$ ratios
of trend slopes between each pairs of series given by%
\[
\theta^{(i)}=\frac{\beta_{1}^{(i)}}{\beta_{2}^{(i)}},
\]
where $\beta_{2}^{(i)}\neq0$. Using simple algebra from
\cite{vogelsang-nawaz-JTSA}, estimating equations for the ratios can be
derived as%
\begin{equation}
y_{1t}^{(i)}=\delta^{(i)}+\theta^{(i)}y_{2t}^{(i)}+\epsilon_{\theta t}^{(i)},
\label{y1y2reg}%
\end{equation}
where
\[
\delta^{(i)}=\mu_{1}^{(i)}-\theta^{(i)}\mu_{2}^{(i)},\text{ \ \ \ \ }%
\epsilon_{\theta t}^{(i)}=u_{1t}^{(i)}-\theta^{(i)}u_{2t}^{(i)}.
\]
Throughout the paper, the time series are assumed to be covariance stationary
around their respective linear trends and that sufficient weak dependence
holds so that a functional central limit theorem (FCLT) holds for the
$2n\times1$ vector:%
\[
\mathbf{U}_{t}=\left[
\begin{array}
[c]{c}%
\mathbf{U}_{1t}\\
\mathbf{U}_{2t}%
\end{array}
\right]  ,
\]
where $\mathbf{U}_{1t}=\left[  u_{1t}^{(1)},u_{1t}^{(2)},\ldots,u_{1t}%
^{(n)}\right]  ^{\prime}$ and $\mathbf{U}_{2t}=\left[  u_{2t}^{(1)}%
,u_{2t}^{(2)},\ldots,u_{2t}^{(n)}\right]  ^{\prime}$. Specifically, it is
assumed that%
\begin{equation}
T^{-1/2}%
{\displaystyle\sum_{t=1}^{[rT]}}
\mathbf{U}_{t}=T^{-1/2}%
{\displaystyle\sum_{t=1}^{[rT]}}
\left[
\begin{array}
[c]{c}%
\mathbf{U}_{1t}\\
\mathbf{U}_{2t}%
\end{array}
\right]  \Rightarrow\left[
\begin{array}
[c]{c}%
\mathbf{B}_{\mathbf{u}1}(r)\\
\mathbf{B}_{\mathbf{u}2}(r)
\end{array}
\right]  \equiv\mathbf{B}_{\mathbf{u}}\mathbf{(r),} \label{fclt}%
\end{equation}
where $r\in\lbrack0,1]$ and $[rT]$ is the integer part of $rT$. The elements
of the $n\times1$ vectors of Brownian motions, $\mathbf{B}_{\mathbf{u}1}(r)$
and $\mathbf{B}_{\mathbf{u}2}(r)$, are given by%
\[
\mathbf{B}_{\mathbf{u}1}(r)=\left[  B_{u1}^{(1)}(r),B_{u1}^{(2)}%
(r)\ldots,B_{u1}^{(n)}(r)\right]  ^{\prime},\text{ \ \ \ \ }\mathbf{B}%
_{\mathbf{u}2}(r)=\left[  B_{u2}^{(1)}(r),B_{u2}^{(2)}(r)\ldots,B_{u2}%
^{(n)}(r)\right]  ^{\prime}.
\]
The vector of Brownian motions, $\mathbf{B}_{\mathbf{u}}\mathbf{(r)}$, can be
written as $\mathbf{\Lambda}_{\mathbf{u}}\mathbf{W}_{\mathbf{u}}(r)$ where
$\mathbf{W}_{\mathbf{u}}(r)$ is a $2n\times1$ vector of independent standard
Wiener processes and $\mathbf{\Omega}_{\mathbf{u}}=\mathbf{\Lambda
}_{\mathbf{u}}\mathbf{\Lambda}_{\mathbf{u}}^{\prime}$ is the long run variance
of $\mathbf{U}_{t}$. It is\textit{\ not} assumed that $\mathbf{\Omega
}_{\mathbf{u}}$ is diagonal allowing for correlation across elements
of\ $\mathbf{U}_{t}$ (within and across pairs). In addition to (\ref{fclt}),
it is assumed that $\mathbf{U}_{t}$ is ergodic for the first and second moments.

Stacking the individual estimation equation errors, $\epsilon_{\theta t}%
^{(i)}$, gives the $n\times1$ vector, $\scalebox{1.5}{$\bm{\epsilon}$}_{\theta
t}$, that can be written as%
\[
\scalebox{1.5}{$\bm{\epsilon}$}_{\theta t}=\mathbf{U}_{1t}-\mathbf{D}%
_{\mathbf{\theta}}\mathbf{U}_{2t},
\]
where%
\[
\mathbf{\theta=}\left[  \theta^{(1)},\theta^{(2)},\ldots,\theta^{(n)}\right]
^{\prime},
\]
and $\mathbf{D}_{\mathbf{\theta}}$ is an $n\times n$ diagonal matrix with
$i^{th}$ diagonal elements $\theta^{(i)}$. Notice that
$\scalebox{1.5}{$\bm{\epsilon}$}_{\theta t}$ can be written as a linear
function of $\mathbf{U}_{t}$ through the relationship%
\[
\scalebox{1.5}{$\bm{\epsilon}$}_{\theta t}=\left[  \mathbf{I}_{n}%
,-\mathbf{D}_{\mathbf{\theta}}\right]  \mathbf{U}_{t},
\]
where $\mathbf{I}_{n}$ is an $n\times n$ identity matrix. It immediately
follows from (\ref{fclt}) that%
\begin{equation}
T^{-1/2}%
{\displaystyle\sum_{t=1}^{[rT]}}
\scalebox{1.5}{$\bm{\epsilon}$}_{\theta t}\Rightarrow\left[  \mathbf{I}%
_{n},-\mathbf{D}_{\mathbf{\theta}}\right]  \mathbf{B}_{\mathbf{u}%
}\mathbf{(r)\sim\Lambda}_{\mathbf{\epsilon}}\mathbf{W}_{\mathbf{\epsilon}}(r),
\label{fclteps}%
\end{equation}
where $\mathbf{W}_{\mathbf{\epsilon}}(r)$ is an $n\times1$ vector of
independent Wiener processes and%
\[
\mathbf{\Omega}_{\mathbf{\epsilon}}=\mathbf{\Lambda}_{\mathbf{\epsilon}%
}\mathbf{\Lambda}_{\mathbf{\epsilon}}^{\prime}=\left[  \mathbf{I}%
_{n},-\mathbf{D}_{\mathbf{\theta}}\right]  \mathbf{\Lambda}_{\mathbf{u}%
}\mathbf{\Lambda}_{\mathbf{u}}^{\prime}\left[  \mathbf{I}_{n},-\mathbf{D}%
_{\mathbf{\theta}}\right]  ^{\prime},
\]
is the long run variance of $\scalebox{1.5}{$\bm{\epsilon}$}_{\theta t}$.

\subsection{Estimation of the Trend Slope Ratios}

The trend slope ratios $\theta^{(i)}$ can be estimated using the estimating
equations (\ref{y1y2reg}). While ordinary least squares (OLS) applied equation
by equation for each $i$ would seem natural, \cite{vogelsang-nawaz-JTSA}
showed that OLS applied to (\ref{y1y2reg}), for a given $i$, yields a biased
estimator of $\theta^{(i)}$. The bias is caused by correlation between
$y_{2t}^{(i)}$ and $\epsilon_{\theta t}^{(i)}$ through the common term
$u_{2t}^{(i)}$. Instead, \cite{vogelsang-nawaz-JTSA} recommend using
instrumental variables (IV) estimation using $t$ as the instrument for
$y_{2t}^{(i)}$. The IV estimators are defined as%
\begin{equation}
\widehat{\theta}^{(i)}=\left(
{\displaystyle\sum\limits_{t=1}^{T}}
(t-\overline{t})(y_{2t}^{(i)}-\overline{y}_{2}^{(i)})\right)  ^{-1}%
{\displaystyle\sum\limits_{t=1}^{T}}
(t-\overline{t})(y_{1t}^{(i)}-\overline{y}_{1}^{(i)}), \label{iv}%
\end{equation}
where $\overline{y}_{1}^{(i)}=T^{-1}\sum_{t=1}^{T}y_{1t}^{(i)}$, $\overline
{y}_{2}^{(i)}=T^{-1}\sum_{t=1}^{T}y_{2t}^{(i)}$ and $\overline{t}=T^{-1}%
\sum_{t=1}^{T}t$ are sample averages. Standard algebra gives the relationship%
\[
\widehat{\theta}^{(i)}-\theta^{(i)}=\left(
{\displaystyle\sum\limits_{t=1}^{T}}
(t-\overline{t})(y_{2t}^{(i)}-\overline{y}_{2}^{(i)})\right)  ^{-1}%
{\displaystyle\sum\limits_{t=1}^{T}}
(t-\overline{t})\epsilon_{\theta t}^{(i)}.
\]
Notice that the IV estimator can be equivalently written as%
\[
\widehat{\theta}^{(i)}=\frac{\widehat{\beta}_{1}^{(i)}}{\widehat{\beta}%
_{2}^{(i)}},
\]
where $\widehat{\beta}_{1}^{(i)}$ and $\widehat{\beta}_{2}^{(i)}$ are the OLS
estimators of $\beta_{1}^{(i)}$ and $\beta_{2}^{(i)}$ based on regressions
(\ref{1.1}) and (\ref{1.2}):%
\begin{equation}
\widehat{\beta}_{1}^{(i)}=\left(  \sum_{t=1}^{T}(t-\overline{t})^{2}\right)
^{-1}\sum_{t=1}^{T}(t-\overline{t})(y_{1t}^{(i)}-\overline{y}_{1}^{(i)}),
\label{beta1ols}%
\end{equation}%
\begin{equation}
\widehat{\beta}_{2}^{(i)}=\left(  \sum_{t=1}^{T}(t-\overline{t})^{2}\right)
^{-1}\sum_{t=1}^{T}(t-\overline{t})(y_{2t}^{(i)}-\overline{y}_{2}^{(i)}).
\label{beta2ols}%
\end{equation}
Stack the $\widehat{\theta}^{(i)}$ into the vector%
\[
\widehat{\mathbf{\theta}}\mathbf{=}\left[  \widehat{\theta}^{(1)}%
,\widehat{\theta}^{(2)},\ldots,\widehat{\theta}^{(n)}\right]  ^{\prime}.
\]
Using (\ref{iv}), $\widehat{\mathbf{\theta}}$ can be written as%
\[
\widehat{\mathbf{\theta}}\text{ }\mathbf{=}\text{ }\widehat{\mathbf{D}}%
_{2}^{-1}\sum_{t=1}^{T}\left(  \mathbf{y}_{1t}-\overline{\mathbf{y}}%
_{1}\right)  \left(  t-\overline{t}\right)  ,
\]
where $\mathbf{y}_{1t}=\left[  y_{1t}^{(1)},y_{1t}^{(2)},\ldots,y_{1t}%
^{(n)}\right]  ^{\prime}$ and $\widehat{\mathbf{D}}_{2}$ is an $n\times n$
diagonal matrix with $i^{th}$ diagonal element given by $\sum_{t=1}^{T}\left(
t-\overline{t}\right)  \left(  y_{2t}^{(i)}-\overline{y}_{2}^{(i)}\right)  $.
Standard calculations give%
\begin{equation}
\widehat{\mathbf{\theta}}-\mathbf{\theta=}\text{ }\widehat{\mathbf{D}}%
_{2}^{-1}\sum_{t=1}^{T}\left(  t-\overline{t}\right)
\scalebox{1.5}{$\bm{\epsilon}$}_{\theta t}, \label{thetahat_centered}%
\end{equation}
which holds as long as the trend slopes are nonzero. When trends slopes are
zero, $\mathbf{\theta}$ is not defined and it follows that%
\begin{equation}
\widehat{\theta}^{(i)}=\frac{%
{\textstyle\sum\nolimits_{t=1}^{T}}
(t-\overline{t})(u_{1t}^{(i)}-\overline{u}_{1}^{(i)})}{%
{\textstyle\sum\nolimits_{t=1}^{T}}
(t-\overline{t})(u_{2t}^{(i)}-\overline{u}_{2}^{(i)})}=\frac{%
{\textstyle\sum\nolimits_{t=1}^{T}}
(t-\overline{t})u_{1t}^{(i)}}{%
{\textstyle\sum\nolimits_{t=1}^{T}}
(t-\overline{t})u_{2t}^{(i)}}. \label{thetahat_zero_slopes}%
\end{equation}
For the rest of the paper, the focus is on the IV estimator,
$\widehat{\mathbf{\theta}}$, and tests of linear restrictions regarding
$\mathbf{\theta}$.

\subsection{Asymptotic Properties of the IV Estimator}

\noindent The following Theorem gives the asymptotic properties of
$\widehat{\mathbf{\theta}}$ under the assumption that the FCLT (\ref{fclt})
holds. The asymptotic limit of $\widehat{\mathbf{\theta}}$ depends on the
magnitude of the trend slope parameters, $\beta_{1}^{(i)},\beta_{2}^{(i)}$
relative to the variation in the random components (noise), $u_{1t}^{(i)}$ and
$u_{2t}^{(i)}$.

\begin{theorem}
Suppose that (\ref{fclt}) holds which implies that (\ref{fclteps}) holds. Let
$\overline{\beta}_{1}^{(i)},\overline{\beta}_{2}^{(i)}$ be fixed with respect
to $T$. Let $\mathbf{D}_{\overline{\mathbf{\beta}}_{2}}$ be an $n\times n$
diagonal matrix with $i^{th}$ diagonal element $\overline{\beta}_{2}^{(i)}$,
and let $\mathbf{D}_{\mathbf{B}_{u2}}$be an $n\times n$ diagonal matrix with
$i^{th}$ diagonal element $\int_{0}^{1}\left(  s-\frac{1}{2}\right)
dB_{u2}^{(i)}(s)$. The following hold as $T\rightarrow\infty$:%
\[
\]
Case 1 (large to small slopes): For $\beta_{1}^{(i)}=T^{-\kappa}%
\overline{\beta}_{1}^{(i)}$, $\beta_{2}^{(i)}=T^{-\kappa}\overline{\beta}%
_{2}^{(i)}$ with $0\leq\kappa<\frac{3}{2}$,%
\[
T^{3/2-\kappa}\left(  \widehat{\mathbf{\theta}}-\mathbf{\theta}\right)
\Rightarrow\left(  \frac{1}{12}\mathbf{D}_{\overline{\mathbf{\beta}}_{2}%
}\right)  ^{-1}\mathbf{\Lambda}_{\mathbf{\epsilon}}\int_{0}^{1}\left(
s-\frac{1}{2}\right)  d\mathbf{W}_{\mathbf{\epsilon}}(s)\sim N\left(
\mathbf{0},12\mathbf{D}_{\overline{\mathbf{\beta}}_{2}}^{-1}\mathbf{\Omega
}_{\mathbf{\epsilon}}\mathbf{D}_{\overline{\mathbf{\beta}}_{2}}^{-1}\right)
,
\]
Case 2 (very small slopes): For $\beta_{1}^{(i)}=T^{-3/2}\overline{\beta}%
_{1}^{(i)},$ $\beta_{2}^{(i)}=T^{-3/2}\overline{\beta}_{2}^{(i)}$
($\kappa=\frac{3}{2}$),%
\[
T^{3/2}\left(  \widehat{\mathbf{\theta}}-\mathbf{\theta}\right)  =\left(
\widehat{\mathbf{\theta}}-\mathbf{\theta}\right)  \Rightarrow\left(  \frac
{1}{12}\mathbf{D}_{\overline{\mathbf{\beta}}_{2}}+\mathbf{D}_{\mathbf{B}_{u2}%
}\right)  ^{-1}\mathbf{\Lambda}_{\mathbf{\epsilon}}\int_{0}^{1}\left(
s-\frac{1}{2}\right)  d\mathbf{W}_{\mathbf{\epsilon}}(s).
\]
Case 3 (zero slopes): For $\beta_{1}^{(i)}=0$, $\beta_{2}^{(i)}=0$,%
\[
\widehat{\mathbf{\theta}}\Rightarrow\mathbf{D}_{\mathbf{B}_{u2}}^{-1}\int%
_{0}^{1}\left(  s-\frac{1}{2}\right)  d\mathbf{B}_{\mathbf{u}1}(s).
\]

\end{theorem}

For Cases 1 and 2 the limits given in Theorem 1 are multivariate versions of
the limits obtained by \cite{vogelsang-nawaz-JTSA} and are identical to the
limits in \cite{vogelsang-nawaz-JTSA} when $n=1$. In Case 1,
$\widehat{\mathbf{\theta}}$ consistently estimates $\mathbf{\theta}$ and the
precision of $\widehat{\mathbf{\theta}}$ depends on the magnitudes of the
trends slopes in the denominator series. In Case 2, $\widehat{\mathbf{\theta}%
}$ becomes inconsistent. This is not surprising because the case of very small
slopes means that the trend component of $y_{2t}^{(i)}$ is dominated by the
noise, $u_{2t}^{(i)}$, in which case $t$ is a weak instrument (Staiger and
Stock 1997\nocite{staiger-stock}) for $y_{2t}^{(i)}$. When all trend slopes
are zero (Case 3), trend ratios are not defined and $\widehat{\mathbf{\theta}%
}$ converges to a random vector that depends on $\mathbf{B}_{\mathbf{u}%
}\mathbf{(r)}$. It is obvious that for very small and zero slopes, inference
will be affected by the different behavior of $\widehat{\mathbf{\theta}}$.

\section{Testing Linear Restrictions Across Trend Ratios}

\noindent Suppose one is interested in testing linear restrictions across the
trend slopes, $\theta^{(i)}$, using the IV estimators $\widehat{\theta}^{(i)}%
$. More formally, consider testing the null hypothesis%
\[
H_{0}:\mathbf{R\theta=r,}%
\]
against the alternative%
\[
H_{1}:\mathbf{R\theta\neq r,}%
\]
where $\mathbf{R}$ is a known $q\times n$ matrix with $rank(\mathbf{R)=}q$ and
$\mathbf{r}$ is a known $q\times1$ vector.

Using the Case 1 limit from Theorem 1, the Wald statistic for testing $H_{0}$
against $H_{1}$ is given by%
\[
Wald_{IV}=\left(  \mathbf{R}\widehat{\mathbf{\theta}}-\mathbf{r}\right)
^{\prime}\left[  \mathbf{R}\widehat{\mathbf{V}}_{IV}\mathbf{R}^{\prime
}\right]  ^{-1}\left(  \mathbf{R}\widehat{\mathbf{\theta}}-\mathbf{r}\right)
\]
where%
\[
\widehat{\mathbf{V}}_{IV}=\left(  \sum\nolimits_{t=1}^{T}(t-\overline{t}%
)^{2}\right)  \widehat{\mathbf{D}}_{2}^{-1}\widehat{\mathbf{\Omega}%
}_{\mathbf{\epsilon}}\widehat{\mathbf{D}}_{2}^{-1}.
\]
The middle term of $\widehat{\mathbf{V}}_{IV}$ is an estimator of
$\mathbf{\Omega}_{\mathbf{\epsilon}}$ given by%
\[
\widehat{\mathbf{\Omega}}_{\mathbf{\epsilon}}=\widehat{\mathbf{\Gamma}%
}_{\mathbf{\epsilon}0}+%
{\displaystyle\sum\limits_{j=1}^{T-1}}
k\left(  \frac{j}{M}\right)  \left(  \widehat{\mathbf{\Gamma}}%
_{\mathbf{\epsilon}j}+\widehat{\mathbf{\Gamma}}_{\mathbf{\epsilon}j}^{\prime
}\right)  \text{, \ \ }\widehat{\mathbf{\Gamma}}_{\mathbf{\epsilon}j}=T^{-1}%
{\displaystyle\sum_{t=j+1}^{T}}
\widehat{\scalebox{1.5}{$\bm{\epsilon}$}}_{\theta t}%
\widehat{\scalebox{1.5}{$\bm{\epsilon}$}}_{\theta t-j}^{\prime},
\]
where $\widehat{\scalebox{1.5}{$\bm{\epsilon}$}}_{\theta t}$ is the vector of
IV residuals given by%
\[
\widehat{\scalebox{1.5}{$\bm{\epsilon}$}}_{\theta t}=\left[  \widehat{\epsilon
}_{\theta t}^{(1)},\widehat{\epsilon}_{\theta t}^{(2)},\ldots
,\widehat{\epsilon}_{t}^{(n)}\right]  ^{\prime},
\]
with%
\[
\widehat{\epsilon}_{t}^{(i)}=y_{1t}^{(i)}-\overline{y}_{1}^{(i)}%
-\widehat{\theta}^{(i)}\left(  y_{2t}^{(i)}-\overline{y}_{2}^{(i)}\right)  .
\]
The long run variance estimator, $\widehat{\mathbf{\Omega}}_{\mathbf{\epsilon
}}$, is of the well known kernel form where $k(x)$ is the kernel
(downweighting function) and $M$ is the bandwidth tuning parameter that
controls the extent of downweighting by the kernel. For the case where one
restriction is being tested, $q=1$, a $t$-statistic can be used to test
one-sided hypotheses:%
\[
t_{IV}=\frac{\mathbf{R}\widehat{\mathbf{\theta}}-\mathbf{r}}{\sqrt
{\mathbf{R}\widehat{\mathbf{V}}_{IV}\mathbf{R}^{\prime}}}.
\]

In order to understand the power properties of $Wald_{IV}$ and $t_{IV}$, their
asymptotic limits are derived for local alternatives, $H_{1L}$, of the form%
\[
H_{1L}:\mathbf{R\theta=r+}\overline{\mathbf{\Delta}}T^{-3/2+\kappa}\mathbf{,}%
\]
where $\kappa$ is the same parameter used in Theorem 1 to model the trend
slopes as local to zero. In deriving the asymptotic results, the bandwidth
parameter for $\widehat{\mathbf{\Omega}}_{\mathbf{\epsilon}}$ is assumed to be
a fixed proportion of the sample size, $b\in(0,1]$, i.e. $M=bT$. Modeling
$M/T$ as a fixed constant gives the fixed-$b$ (or fixed-smoothing) limit of
$\widehat{\mathbf{\Omega}}_{\mathbf{\epsilon}}$ (and $t_{IV}$). The advantage
of the fixed-$b$ approach is that it delivers an asymptotic random variable
and associated critical values that depend on the bandwidth and kernel. This
is in contrast to appealing to a consistency result for
$\widehat{\mathbf{\Omega}}_{\mathbf{\epsilon}}$ which would not depend on the
bandwidth or kernel. For more details on the fixed-$b$ approach see
\cite{bunzel-vogelsang}, \cite{bunzel-vogelsang}, \cite{jansson-kvb},
\cite{kvnewasym}, \cite{phillips-sun-jin-higherpower}, \cite{zhang2013fixed},
\cite{sun-twostepgmm}, \cite{LLSW18}, and \cite{LLS21}.

Because the form of the fixed-$b$ limit of the test statistics depends on the
type of kernel function used to compute $\widehat{\mathbf{\Omega}%
}_{\mathbf{\epsilon}}$, definitions from \cite{kvnewasym} are used. A kernel
is labelled Type 1 if $k\left(  x\right)  $ is twice continuously
differentiable everywhere and as a Type 2 kernel if $k\left(  x\right)  $ is
continuous, $k\left(  x\right)  =0$ for $\left\vert x\right\vert \geq1$ and
$k\left(  x\right)  $ is twice continuously differentiable everywhere except
at $\left\vert x\right\vert =1.$ The Bartlett kernel (which is neither Type 1
or 2) is considered separately.

The fixed-$b$ limiting distributions are expressed in terms of the following
stochastic functions. Let $\mathbf{Q}(r)$ be a generic vector stochastic
process. Define the random variable $\mathbf{P}_{b}(\mathbf{Q}(r))$ as%
\[
\mathbf{P}_{b}(\mathbf{Q}(r))=\left\{
\begin{array}
[c]{cc}%
\int_{0}^{1}\int_{0}^{1}-k^{\ast\prime\prime}\left(  r-s\right)
\mathbf{Q}(r)\mathbf{Q}(s)^{\prime}drds & \text{if }k\left(  x\right)  \text{
is Type 1}\\
& \\
\int\int_{\left\vert r-s\right\vert <b}-k^{\ast\prime\prime}\left(
r-s\right)  \mathbf{Q}\left(  r\right)  \mathbf{Q}\left(  s\right)  ^{\prime
}drds & \\
+k_{-}^{\ast\prime}\left(  b\right)  \int_{0}^{1-b}\left(  \mathbf{Q}\left(
r+b\right)  \mathbf{Q}\left(  r\right)  ^{\prime}+\mathbf{Q}\left(  r\right)
\mathbf{Q}\left(  r+b\right)  ^{\prime}\right)  dr & \text{if }k\left(
x\right)  \text{ is Type 2}\\
& \\
\frac{2}{b}\int_{0}^{1}\mathbf{Q}\left(  r\right)  \mathbf{Q}\left(  r\right)
^{\prime}dr-\frac{1}{b}\int_{0}^{1-b}\left(  \mathbf{Q}\left(  r+b\right)
\mathbf{Q}\left(  r\right)  ^{\prime}+\mathbf{Q}\left(  r\right)
\mathbf{Q}\left(  r+b\right)  ^{\prime}\right)  dr & \text{if }k\left(
x\right)  \text{ is Bartlett}%
\end{array}
\right.
\]
where $k^{\ast}(x)=k\left(  \frac{x}{b}\right)  $ and $k_{-}^{\ast^{\prime}}$
is the first derivative of $k^{\ast}$ from below (left).

\begin{theorem}
Suppose that (\ref{fclt}) holds which implies that (\ref{fclteps}) holds. Let
$\overline{\beta}_{1}^{(i)},\overline{\beta}_{2}^{(i)}$ be fixed with respect
to $T$. Let $\mathbf{D}_{\overline{\mathbf{\beta}}_{2}}$ and $\mathbf{D}%
_{\mathbf{B}_{u2}}$ be defined as in Theorem 1. Let $\mathbf{\Lambda
}_{\mathbf{\epsilon}}^{\ast}$ be the matrix square root of $\mathbf{\Omega
}_{\epsilon}^{\ast}=\mathbf{RA}_{\mathbf{B}_{2}}^{-1}\mathbf{\Lambda
}_{\mathbf{\epsilon}}\mathbf{\Lambda}_{\mathbf{\epsilon}}^{\prime}%
\mathbf{A}_{\mathbf{B}_{2}}^{-1}\mathbf{R}^{\prime}$ ($\mathbf{\Lambda
}_{\mathbf{\epsilon}}^{\ast}\mathbf{\Lambda}_{\mathbf{\epsilon}}^{\ast\prime
}=\mathbf{\Omega}_{\epsilon}^{\ast}$). Suppose $\mathbf{R\theta=r+}%
\overline{\mathbf{\Delta}}T^{-3/2+\kappa}$. The following hold as
$T\rightarrow\infty$.%
\[
\]
Case 1 (large to small slopes): For $\beta_{1}^{(i)}=T^{-\kappa}%
\overline{\beta}_{1}^{(i)}$, $\beta_{2}^{(i)}=T^{-\kappa}\overline{\beta}%
_{2}^{(i)}$ with $0\leq\kappa<\frac{3}{2}$,%
\[
Wald_{IV}\Rightarrow\left(  \mathbf{Z}_{\epsilon}^{\ast}+\frac{1}{\sqrt{12}%
}\mathbf{\Lambda}_{\mathbf{\epsilon}}^{\ast-1}\overline{\mathbf{\Delta}%
}\right)  ^{\prime}\mathbf{P}_{b}(\widetilde{\mathbf{W}}_{\mathbf{\epsilon}%
}^{\ast}(r))^{-1}\left(  \mathbf{Z}_{\epsilon}^{\ast}+\frac{1}{\sqrt{12}%
}\mathbf{\Lambda}_{\mathbf{\epsilon}}^{\ast-1}\overline{\mathbf{\Delta}%
}\right)  ,
\]
for $q=1$%
\[
t_{IV}\Rightarrow\frac{\mathbf{Z}_{\epsilon}^{\ast}+\frac{1}{\sqrt{12}%
}\mathbf{\Lambda}_{\mathbf{\epsilon}}^{\ast-1}\overline{\mathbf{\Delta}}%
}{\sqrt{\mathbf{P}_{b}(\widetilde{\mathbf{W}}_{\mathbf{\epsilon}}^{\ast}(r))}%
},
\]
where $\mathbf{Z}_{\epsilon}^{\ast}=\sqrt{12}\int_{0}^{1}\left(  s-\frac{1}%
{2}\right)  d\mathbf{W}_{\mathbf{\epsilon}}^{\ast}(s)$, $\widetilde{\mathbf{W}%
}_{\mathbf{\epsilon}}^{\ast}(r)=\mathbf{W}_{\mathbf{\epsilon}}^{\ast
}(r)-r\mathbf{W}_{\mathbf{\epsilon}}^{\ast}(1)-12L(r)\int_{0}^{1}\left(
s-\frac{1}{2}\right)  d\mathbf{W}_{\mathbf{\epsilon}}^{\ast}(s)$,
$L(r)=\int_{0}^{r}\left(  s-\frac{1}{2}\right)  ds$, and $\mathbf{W}%
_{\mathbf{\epsilon}}^{\ast}(r)$ is a $q\times1$ vector of independent Wiener
processes. Note that $\mathbf{Z}_{\epsilon}^{\ast}\sim\mathbf{N}%
(0,\mathbf{I}_{q})$ and is independent of $\widetilde{\mathbf{W}%
}_{\mathbf{\epsilon}}^{\ast}(r)$.%
\[
\]
Case 2 (very small slopes): For $\beta_{1}^{(i)}=T^{-3/2}\overline{\beta}%
_{1}^{(i)},$ $\beta_{2}^{(i)}=T^{-3/2}\overline{\beta}_{2}^{(i)}$
($\kappa=\frac{3}{2}$),%
\begin{align*}
Wald_{IV} &  \Rightarrow\left(  \mathbf{R}\left(  \frac{1}{12}\mathbf{D}%
_{\overline{\mathbf{\beta}}_{2}}+\mathbf{D}_{\mathbf{B}_{u2}}\right)
^{-1}\mathbf{\Lambda}_{\mathbf{\epsilon}}\int_{0}^{1}\left(  s-\frac{1}%
{2}\right)  d\mathbf{W}_{\mathbf{\epsilon}}(s)+\overline{\mathbf{\Delta}%
}\right)  ^{\prime}\\
&  \times\left[  \mathbf{R}\left(  \frac{1}{12}\mathbf{D}_{\overline
{\mathbf{\beta}}_{2}}+\mathbf{D}_{\mathbf{B}_{u2}}\right)  ^{-1}\mathbf{P}%
_{b}(\mathbf{H}_{1}(r))\left(  \frac{1}{12}\mathbf{D}_{\overline
{\mathbf{\beta}}_{2}}+\mathbf{D}_{\mathbf{B}_{u2}}\right)  ^{-1}%
\mathbf{R}^{\prime}\right]  ^{-1}\\
&  \times\left(  \mathbf{R}\left(  \frac{1}{12}\mathbf{D}_{\overline
{\mathbf{\beta}}_{2}}+\mathbf{D}_{\mathbf{B}_{u2}}\right)  ^{-1}%
\mathbf{\Lambda}_{\mathbf{\epsilon}}\int_{0}^{1}\left(  s-\frac{1}{2}\right)
d\mathbf{W}_{\mathbf{\epsilon}}(s)+\overline{\mathbf{\Delta}}\right)  ,
\end{align*}
for $q=1$%
\[
t_{IV}\Rightarrow\frac{\mathbf{R}\left(  \frac{1}{12}\mathbf{D}_{\overline
{\mathbf{\beta}}_{2}}+\mathbf{D}_{\mathbf{B}_{u2}}\right)  ^{-1}%
\mathbf{\Lambda}_{\mathbf{\epsilon}}\int_{0}^{1}\left(  s-\frac{1}{2}\right)
d\mathbf{W}_{\mathbf{\epsilon}}(s)+\overline{\mathbf{\Delta}}}{\sqrt{\frac
{1}{12}\mathbf{R}\left(  \frac{1}{12}\mathbf{D}_{\overline{\mathbf{\beta}}%
_{2}}+\mathbf{D}_{\mathbf{B}_{u2}}\right)  ^{-1}\mathbf{P}_{b}(\mathbf{H}%
_{1}(r))\left(  \frac{1}{12}\mathbf{D}_{\overline{\mathbf{\beta}}_{2}%
}+\mathbf{D}_{\mathbf{B}_{u2}}\right)  ^{-1}\mathbf{R}^{\prime}}},
\]
where $\mathbf{D}_{\overline{\mathbf{\beta}}_{2}}$ and $\mathbf{D}%
_{\mathbf{B}_{u2}}$ are defined in Theorem 1,%
\[
\mathbf{H}_{1}(r)=\mathbf{\Lambda}_{\mathbf{\epsilon}}\left(  \mathbf{W}%
_{\mathbf{\epsilon}}(r)-r\mathbf{W}_{\mathbf{\epsilon}}(1)\right)  -\left(
L(r)\mathbf{D}_{\overline{\mathbf{\beta}}_{2}}+\mathbf{D}_{\widehat{\mathbf{B}%
}_{u2}(r)}\right)  \left(  \left(  \frac{1}{12}\mathbf{D}_{\overline
{\mathbf{\beta}}_{2}}+\mathbf{D}_{\mathbf{B}_{u2}}\right)  ^{-1}%
\mathbf{\Lambda}_{\mathbf{\epsilon}}\int_{0}^{1}\left(  s-\frac{1}{2}\right)
d\mathbf{W}_{\mathbf{\epsilon}}(s)\right)  ,
\]
and $\mathbf{D}_{\widehat{\mathbf{B}}_{u2}(r)}$ is an $n\times n$ diagonal
matrix with $i^{th}$ diagonal element $\widehat{B}_{u2}^{(i)}(r)=B_{u2}%
^{(i)}(r)-rB_{u2}^{(i)}(1)$.%
\[
\]
Case 3 (zero slopes): For $\beta_{1}^{(i)}=0$, $\beta_{2}^{(i)}=0$,%
\begin{align*}
Wald_{IV} &  \Rightarrow\left(  \mathbf{RD}_{\mathbf{B}_{u2}}^{-1}\int_{0}%
^{1}\left(  s-\frac{1}{2}\right)  d\mathbf{B}_{\mathbf{u}1}(s)-\mathbf{r}%
\right)  ^{\prime}\left[  \frac{1}{12}\mathbf{RD_{\mathbf{B}_{u2}}%
^{-1}\mathbf{P}_{b}(\mathbf{H}_{2}(r))D_{\mathbf{B}_{u2}}^{-1}R}^{\prime
}\right]  ^{-1}\\
&  \text{ \ \ \ \ \ \ \ \ \ \ \ \ \ \ \ \ \ \ \ \ }\times\left(
\mathbf{RD}_{\mathbf{B}_{u2}}^{-1}\int_{0}^{1}\left(  s-\frac{1}{2}\right)
d\mathbf{B}_{\mathbf{u}1}(s)-\mathbf{r}\right)  ,
\end{align*}
for $q=1$%
\[
t_{IV}\Rightarrow\frac{\mathbf{RD}_{\mathbf{B}_{u2}}^{-1}\int_{0}^{1}\left(
s-\frac{1}{2}\right)  d\mathbf{B}_{\mathbf{u}1}(s)-\mathbf{r}}{\sqrt{\frac
{1}{12}\mathbf{RD_{\mathbf{B}_{u2}}^{-1}P}_{b}(\mathbf{H}_{2}%
(r))\mathbf{D_{\mathbf{B}_{u2}}^{-1}R}^{\prime}}},
\]
where%
\[
\mathbf{H}_{2}(r)=\mathbf{B}_{\mathbf{u}1}(s)-r\mathbf{B}_{\mathbf{u}%
1}(1)-\mathbf{D}_{\widehat{\mathbf{B}}_{u2}(r)}\mathbf{D}_{\mathbf{B}_{u2}%
}^{-1}\int_{0}^{1}\left(  s-\frac{1}{2}\right)  d\mathbf{B}_{\mathbf{u}1}(s).
\]

\end{theorem}

The limit of $t_{IV}$ given by Case 1 of Theorem 2 is of the same form as the
limit derived by \cite{vogelsang-nawaz-JTSA} for case of a single trend ratio
($n=1$). Under the null hypothesis (when $\overline{\mathbf{\Delta}%
}=\mathbf{0}$\textbf{)}, the limits of $Wald_{IV}$ and $t_{IV}$ are identical
to limits obtained by \cite{bunzel-vogelsang} in deterministic trend
regression models. While the limiting random variables are nonstandard because
of the fixed-$b$ limit of $\widehat{\mathbf{\Omega}}_{\mathbf{\epsilon}}$, the
$\mathbf{P}_{b}(\widetilde{\mathbf{W}}_{\mathbf{\epsilon}}^{\ast}(r))$ term in
the denominator, critical values are available using formulas from
\cite{bunzel-vogelsang}. The critical values depend on the bandwidth through
$b=M/T$ and the kernel, $k(x)$.

When the trend slopes are very small or zero, the limiting distributions
change, become more complicated, and depend on nuisance parameters. This is to
be expected given the weak instrument problem that occurs when the $\beta
_{2}^{(i)}$ slopes are small and the fact that $\mathbf{\theta}$ is not
defined when slopes are zero. The finite sample simulations will illustrate
the extent to which inference breaks down as the trend slopes become very
small or zero.

\section{Product Approach for Testing Equal Trend Ratios}

When testing a simple hypothesis about a single trend slope ratio,
\cite{vogelsang-nawaz-JTSA} used Fieller's method (\cite{fieller1954}) to
construct confidence intervals that are robust to very small trends slopes
including the case of zero trend slopes. This approach is based on rewriting a
simple hypothesis about $\theta^{(i)}$ in terms of a linear restriction
involving $\beta_{1}^{(i)}$ and $\beta_{2}^{(i)}$. However, once a null
hypothesis involves a linear combination of at least two $\theta^{(i)}$,
Fieller's method cannot be applied.

One important empirical application is testing equality of trend slope ratios
for two pairs of series. See VMCS26 for tests of equal trend ratios between
observed and model generated temperature series. For the null hypothesis of
equal trend ratios, it is possible to develop a testing approach similar in
spirit to Fieller's method that gives potentially more robust inference when
slopes are very small. Suppose there are two pairs of series with trend ratios
$\theta^{(1)}$ and $\theta^{(2)}$ and the hypothesis of interest is%
\[
H_{0}:\theta^{(1)}=\theta^{(2)},
\]
or equivalently%
\begin{equation}
H_{0}:\theta^{(1)}-\theta^{(2)}=0. \label{h0equaltheta}%
\end{equation}
Anticipating a local asymptotic calculation, suppose the alternative is
specified as local to zero%
\[
H_{1L}:\theta^{(1)}-\theta^{(2)}=\overline{\mathbf{\Delta}}T^{-3/2+\kappa}.
\]
Rewriting the alternative hypothesis in terms of the trend slopes gives%
\[
H_{1L}:\frac{\beta_{1}^{(1)}}{\beta_{2}^{(1)}}-\frac{\beta_{1}^{(2)}}%
{\beta_{2}^{(2)}}=\overline{\mathbf{\Delta}}T^{-3/2+\kappa}.
\]
Multiplying both sides of $H_{1L}$ by $\beta_{2}^{(1)}\beta_{2}^{(2)}$ gives%
\begin{equation}
H_{1L}:\beta_{2}^{(2)}\beta_{1}^{(1)}-\beta_{2}^{(1)}\beta_{1}^{(2)}=\beta
_{2}^{(1)}\beta_{2}^{(2)}\overline{\mathbf{\Delta}}T^{-3/2+\kappa}.
\label{h1prod}%
\end{equation}
Testing (\ref{h0equaltheta}) is equivalent to testing%
\begin{equation}
H_{0}:\beta_{2}^{(2)}\beta_{1}^{(1)}-\beta_{2}^{(1)}\beta_{1}^{(2)}=0.
\label{h0prod}%
\end{equation}
The advantage of using (\ref{h0prod}) is that the trend slopes can be directly
estimated by OLS and ratios are avoided.

To develop a test statistic for testing (\ref{h0prod}) against (\ref{h1prod})
define%
\[
g_{\beta}=\beta_{2}^{(2)}\beta_{1}^{(1)}-\beta_{2}^{(1)}\beta_{1}^{(2)}.
\]
The natural estimator of $g_{\beta}$ is given by%
\[
g_{\widehat{\beta}}=\widehat{\beta}_{2}^{(2)}\widehat{\beta}_{1}%
^{(1)}-\widehat{\beta}_{2}^{(1)}\widehat{\beta}_{1}^{(2)},
\]
where $\widehat{\beta}_{1}^{(i)}$ and $\widehat{\beta}_{2}^{(i)}$ are the OLS
estimators (\ref{beta1ols}) and (\ref{beta2ols}).

Because $g_{\widehat{\beta}}$ is a nonlinear function of the slope estimators,
its asymptotic variance depends on $\mathbf{\Omega}_{\mathbf{u}}$, the long
run variance of $\mathbf{U}_{t}=\left[  u_{1t}^{(1)},u_{1t}^{(2)},u_{2t}%
^{(1)},u_{2t\mathbf{U}_{t}}^{(2)}\right]  ^{\prime}$, and the vector%
\[
\mathbf{R}_{\beta}=\left[  \beta_{2}^{(2)},-\beta_{2}^{(1)},-\beta_{1}%
^{(2)},\beta_{1}^{(1)}\right]  ,
\]
which can be derived using the delta method or, as in the appendix, directly.
The feasible version of $\mathbf{R}_{\beta}$ is given by%
\[
\mathbf{R}_{\widehat{\beta}}=\left[  \widehat{\beta}_{2}^{(2)},-\widehat{\beta
}_{2}^{(1)},-\widehat{\beta}_{1}^{(2)},\widehat{\beta}_{1}^{(1)}\right]  .
\]
Let $\widehat{\mathbf{U}}_{t}=\left[  \widehat{u}_{1t}^{(1)},\widehat{u}%
_{1t}^{(2)},\widehat{u}_{2t}^{(1)},\widehat{u}_{2t}^{(2)}\right]  ^{\prime}$
where $\widehat{u}_{1t}^{(i)},\widehat{u}_{2t}^{(i)}$ are the residuals from
(\ref{1.1}) and (\ref{1.2}) estimated by OLS. Define the long run variance
estimator of $\mathbf{\Omega}_{\mathbf{u}}$ as%
\[
\widehat{\mathbf{\Omega}}_{\mathbf{u}}=\widehat{\mathbf{\Gamma}}_{\mathbf{u}%
0}+%
{\displaystyle\sum\limits_{j=1}^{T-1}}
k\left(  \frac{j}{M}\right)  (\widehat{\mathbf{\Gamma}}_{\mathbf{u}%
j}+\widehat{\mathbf{\Gamma}}_{\mathbf{u}j}^{\prime}),\text{ \ \ \ }%
\widehat{\mathbf{\Gamma}}_{\mathbf{u}j}=T^{-1}%
{\displaystyle\sum\limits_{t=j+1}^{T}}
\widehat{\mathbf{U}}_{t}\widehat{\mathbf{U}}_{t-j}^{\prime}.
\]
Using%
\[
\widehat{\lambda}_{g}^{2}=\mathbf{R}_{\widehat{\beta}}\widehat{\mathbf{\Omega
}}_{\mathbf{u}}\mathbf{R}_{\widehat{\beta}}^{\prime},
\]
a $t$-statistic for testing (\ref{h0prod}) can be constructed as%
\[
t_{prod}=\frac{g_{\widehat{\beta}}}{\sqrt{\widehat{\lambda}_{g}^{2}\left(
{\textstyle\sum\nolimits_{t=1}^{T}}
(t-\overline{t})^{2}\right)  ^{-1}}}.
\]
The following Theorem gives the limiting distribution of $t_{prod}$ under the
local alternative (\ref{h1prod}).

\begin{theorem}
Suppose that (\ref{fclt}) holds. Let $\overline{\beta}_{1}^{(i)}%
,\overline{\beta}_{2}^{(i)}$ be fixed with respect to $T$. Let $\mathbf{D}%
_{\overline{\mathbf{\beta}}_{2}}$and $\mathbf{D}_{\mathbf{B}_{u2}}$be defined
as in Theorem 1 for the case of $n=2$. Suppose $\theta^{(1)}-\theta
^{(2)}=\overline{\mathbf{\Delta}}T^{-3/2+\kappa}$. The following hold as
$T\rightarrow\infty$:%
\[
\]
Case 1 (large to small slopes): For $\beta_{1}^{(i)}=T^{-\kappa}%
\overline{\beta}_{1}^{(i)}$, $\beta_{2}^{(i)}=T^{-\kappa}\overline{\beta}%
_{2}^{(i)}$ with $0\leq\kappa<\frac{3}{2}$,%
\[
t_{prod}\Rightarrow\frac{z_{u}^{\ast}+\frac{\overline{\beta}_{2}%
^{(1)}\overline{\beta}_{2}^{(2)}\overline{\mathbf{\Delta}}}{\Lambda_{u}^{\ast
}\sqrt{12}}}{\sqrt{P_{b}(\widetilde{w}_{u}^{\ast}(r))}},
\]
where $z_{u}^{\ast}=\sqrt{12}\int_{0}^{1}\left(  s-\frac{1}{2}\right)
dw_{u}^{\ast}(s)$, $\widetilde{w}_{u}^{\ast}(r)=\widetilde{w}_{u}^{\ast
}(r)-rw_{u}^{\ast}(1)-12L(r)\int_{0}^{1}\left(  s-\frac{1}{2}\right)
dw_{u}^{\ast}(s)$, $L(r)=\int_{0}^{r}\left(  s-\frac{1}{2}\right)  ds$,
$w_{u}^{\ast}(r)$ is a standard Wiener process, and $\Lambda_{u}^{\ast}%
=\sqrt{\mathbf{R}_{\overline{\beta}}\mathbf{\Omega}_{\mathbf{u}}%
\mathbf{R}_{\overline{\beta}}^{\prime}}$. Note that $z_{u}^{\ast}\sim N(0,1)$
and is independent of $\widetilde{w}_{u}^{\ast}(r)$.%
\[
\]
Case 2 (very small slopes): For $\beta_{1}^{(i)}=T^{-3/2}\overline{\beta}%
_{1}^{(i)},$ $\beta_{2}^{(i)}=T^{-3/2}\overline{\beta}_{2}^{(i)}$
($\kappa=\frac{3}{2}$),%
\[
t_{prod}\Rightarrow\frac{\mathbf{R}_{\overline{\beta}}\Psi+\Psi_{2}^{(2)}%
\Psi_{1}^{(1)}-\Psi_{2}^{(1)}\Psi_{1}^{(2)}+\overline{\mathbf{\Delta}%
}\overline{\beta}_{2}^{(1)}\overline{\beta}_{2}^{(2)}}{\sqrt{12\left(
\mathbf{R}_{\overline{\beta}}+\left[  \Psi_{2}^{(2)},-\Psi_{2}^{(1)},-\Psi
_{1}^{(2)},\Psi_{1}^{(1)}\right]  \right)  \mathbf{P}_{b}%
(\widetilde{\mathbf{B}}_{\mathbf{u}}(r))\left(  \mathbf{R}_{\overline{\beta}%
}+\left[  \Psi_{2}^{(2)},-\Psi_{2}^{(1)},-\Psi_{1}^{(2)},\Psi_{1}%
^{(1)}\right]  \right)  ^{\prime}}},
\]
where $\mathbf{\Psi}=\left[  \Psi_{1}^{(1)},\Psi_{1}^{(2)},\Psi_{2}^{(1)}%
,\Psi_{2}^{(2)}\right]  ^{\prime}=12\int_{0}^{1}\left(  s-\frac{1}{2}\right)
d\mathbf{B}_{\mathbf{u}}(s)$, $\widetilde{\mathbf{B}}_{\mathbf{u}%
}(r)=\mathbf{B}_{\mathbf{u}}(r)-r\mathbf{B}_{\mathbf{u}}(1)-12L(r)\int_{0}%
^{1}\left(  s-\frac{1}{2}\right)  d\mathbf{B}_{\mathbf{u}}(s).$%
\[
\]
Case 3 (zero slopes): For $\beta_{1}^{(i)}=0$, $\beta_{2}^{(i)}=0$,%
\[
t_{prod}\Rightarrow\frac{\Psi_{2}^{(2)}\Psi_{1}^{(1)}-\Psi_{2}^{(1)}\Psi
_{1}^{(2)}}{\sqrt{12\left[  \Psi_{2}^{(2)},-\Psi_{2}^{(1)},-\Psi_{1}%
^{(2)},\Psi_{1}^{(1)}\right]  \mathbf{P}_{b}(\widetilde{\mathbf{B}%
}_{\mathbf{u}}(r))\left[  \Psi_{2}^{(2)},-\Psi_{2}^{(1)},-\Psi_{1}^{(2)}%
,\Psi_{1}^{(1)}\right]  ^{\prime}}}.
\]

\end{theorem}

When trend slopes are large to small, the asymptotic distribution of
$t_{prod}$ has the same form as $t_{IV}$. The two asymptotic distributions are
the same under the null ($\overline{\mathbf{\Delta}}=0$) but are different
under the local alternative because the variance parameters in the limit are
different. As is the case for $t_{IV}$, the asymptotic distribution of
$t_{prod}$ changes when trends slopes are very small or zero. Given the
complexity of the limits in Cases 2 and 3, it is not clear what Theorem 3
predicts about the finite sample behavior of $t_{prod}$ as trend slopes become
very small or zero. Finite simulations in the next section will be used to
explore the implications of very small or zero trend slopes.

\section{Finite Sample Null Rejection Probabilities and Power for Tests of
Equal Ratios}

This section provides simulation results to show some finite sample properties
of tests of equal trend slopes. The following DGP was used. The $y_{1t}^{(i)}$
and $y_{2t}^{(i)}$ variables where generated by models (\ref{1.1}) and
(\ref{1.2}) where the noise is given by%
\begin{align*}
u_{1t}^{(1)}  &  =\phi_{1}^{(1)}u_{1t-1}^{(1)}+\varepsilon_{1t}^{(1)},\text{
}u_{2t}^{(1)}=\phi_{2}^{(1)}u_{2t-1}^{(1)}+\varepsilon_{2t}^{(1)},\\
u_{1t}^{(2)}  &  =\phi_{1}^{(2)}u_{1t-1}^{(2)}+\varepsilon_{1t}^{(2)},\text{
}u_{2t}^{(2)}=\phi_{2}^{(2)}u_{2t-1}^{(2)}+\varepsilon_{2t}^{(2)},
\end{align*}%
\[
\left[
\begin{array}
[c]{c}%
\varepsilon_{1t}^{(1)}\\
\varepsilon_{2t}^{(1)}\\
\varepsilon_{1t}^{(2)}\\
\varepsilon_{2t}^{(2)}%
\end{array}
\right]  \sim iidN\left(  \mathbf{0},\left[
\begin{array}
[c]{cc}%
\begin{array}
[c]{cc}%
1 & \varphi\\
\varphi & 1
\end{array}
&
\begin{array}
[c]{cc}%
0 & 0\\
0 & 0
\end{array}
\\%
\begin{array}
[c]{cc}%
0 & 0\\
0 & 0
\end{array}
&
\begin{array}
[c]{cc}%
1 & \varphi\\
\varphi & 1
\end{array}
\end{array}
\right]  \right)  ,
\]%
\[
u_{10}^{(1)}=u_{20}^{(1)}=u_{10}^{(2)}=u_{20}^{(2)}=0.
\]
All the noise series are generated by AR(1) processes. Pairs of series are
uncorrelated with each other but there can be within-pair correlation
($\varphi\neq0$).

Given that $t_{IV}$ and $t_{prod}$ are exactly invariant to the values of the
intercept parameters, without loss of generality the intercept parameters are
set to zero: $\mu_{1}^{(1)}=0,\mu_{2}^{(1)}=0,\mu_{1}^{(2)}=0,\mu_{2}^{(2)}%
=0$. Results are reported for various magnitudes of $\beta_{1}^{(i)}$ and
$\beta_{2}^{(i)}$ where $\theta^{(i)}=\beta_{1}^{(i)}/\beta_{2}^{(i)}=1$ under
the null hypothesis of equal ratios across the two pairs. Empirical null
rejections are given for $T=50,100,200$ with $10,000$ replications used in all
cases. Empirical power is given for $T=100$ where $\theta^{(2)}$ takes on
values different from $\theta^{(1)}=1$.

Table 1 reports null rejection probabilities for 5\% nominal level tests for
testing $H_{0}:\theta^{(1)}=\theta^{(2)}$ against the two-sided alternative
$H_{1}:\theta^{(1)}\neq\theta^{(2)}$. Results are reported for the values of
$\beta_{1}^{(i)}=\beta_{2}^{(i)}=10,2,.2,.05,0.25,.005,0$ giving $\theta
^{(i)}=1$ except for the case of $\beta_{1}^{(i)}=\beta_{2}^{(i)}=0$ where
$\theta^{(i)}$ is not defined. The variance estimators use the Daniell kernel.
Results for four bandwidth sample size ratios are provided:
$b=A91,0.25,0.5,1.0$, where $A91$ is the AR(1) plug-in data dependent
bandwidth proposed by \cite{andrews-91}. Empirical rejections are computed
using fixed-$b$ asymptotic critical values with the critical value function%
\[
cv_{0.025}(b)=1.9659+4.0603b+11.6626b^{2}+34.8269b^{3}-13.9506b^{4}%
+3.2669b^{5}\text{,}%
\]
as given by \cite{bunzel-vogelsang} for the Daniell kernel for $k(x)$.

The top panels of Table 1 give results for the case of iid noise as a
benchmark. Here $\phi_{1}^{(1)}=\phi_{2}^{(1)}=\phi_{1}^{(2)}=\phi_{2}%
^{(2)}=\varphi=0$. As long as the trend slopes are large, empirical rejections
are close to the nominal level of 0.05 for both $t_{IV}$ and $t_{prod}$ and
all values of $b$. As the trend slopes get smaller, both tests have rejections
below 0.05 (conservative tests). This happens more quickly with $t_{IV}$ than
$t_{prod}$ and more quickly with smaller values of $b$ (smaller bandwidths).
As $T$ increases, empirical rejections are close to the nominal level for
relatively smaller trend slopes. This makes sense given the asymptotic
results. For given values of $\beta_{1}^{(i)}$ and $\beta_{2}^{(i)}$, the
values of the local trend parameters, $\overline{\beta}_{1}^{(i)}$ and
$\overline{\beta}_{2}^{(i)}$, are larger for bigger values of $T$.

The bottom panels of Table 1 give results for serially correlated noise with
$\phi_{1}^{(1)}=0.3$, $\phi_{2}^{(1)}=0.7$, $\phi_{1}^{(2)}=0.5$, $\phi
_{2}^{(2)}=0.9$. Pairs of series have within-pairs correlation given by
$\varphi=0.5$. The first pair of noise series has modest serial correlation in
the numerator series, $u_{1t}^{(1)}$, and moderate serial correlation in the
denominator series, $u_{2t}^{(1)}$. The second pair of noise series has
stronger serial correlation than the first pair with the denominator series
again having stronger serial correlation than the numerator series. The
relative strengths of serial correlation within pairs are a feature of
temperature series where upper atmosphere temperature (numerators) tend to
have less serial correlation than surface temperatures (denominators). The
relative strength across pairs are a feature of observed versus model
generated temperatures where observed temperatures (first pair) tend to have
less serial correlation than model generated series (second pair).

As in the iid case, empirical rejections in the serial correlation case tend
to fall as trend slopes become smaller indicating that the tests remain
conservative when trend slope are small (or even zero). For larger trend
slopes, the positive autocorrelation can lead to over-rejections especially if
the data-dependent bandwidth is used. As $b$ increases, the tendency to
over-reject is mitigated. This is a well-known property (see
\cite{bunzel-vogelsang} and \cite{kvnewasym} among others). It is not
surprising that over-rejections are larger with the data-dependent bandwidth
because those bandwidths tend to give relatively small values of $b$.
Comparing $t_{IV}$ and $t_{prod}$, one can see that $t_{IV}$ tends to
over-reject less than $t_{prod}$ especially when using the data-dependent
bandwidth. In contrast, $t_{IV}$ tends to under-reject more than $t_{prod}$ as
trend slopes become smaller. The main takeaway is that $t_{IV}$ over-rejects
less than $t_{prod}$ when bandwidths are small and trend slopes are large, but
$t_{IV}$ is more conservative when trend slopes are small (regardless of
bandwidth). As a practical matter, using the slightly larger bandwidth of
$b=0.25$ gives non-trivial reductions in over-rejections relative to the $A91$
bandwidth rule.

Table 2 gives results for empirical power of the tests. Results are only given
for the case of serially correlated noise (patterns are similar for iid noise)
with $T=100$. In all cases $\theta^{(1)}=1$, and results are given for a grid
of values of $\theta^{(2)}$ above and below $1$. The range of the grid
increases as the trend slopes decrease to provide information about the power
curves. It is not surprising that in order to see power, the range of
$\theta^{(2)}$ needs to be increased as trend slopes decrease - this is
predicted by the slower rate of convergence of $\widehat{\mathbf{\theta}}$
(larger sampling variance) as trend slopes become smaller. Empirical power is
not size-adjusted to show actual power in practice. Null rejections are in
bold for the cases where $\theta^{(2)}=1$.

When trend slopes are large, power is similar for $t_{IV}$ and $t_{prod}$ for
the bandwidths $b=0.25,0.5,1.0$. Power is roughly symmetric around null value
of $\theta^{(1)}=1$. When the $A91$ data-dependent bandwidth is used,
$t_{prod}$ over-rejects more than $t_{IV}$ and has correspondingly higher
power. With smaller trend slopes some differences in power emerge between
$t_{IV}$ and $t_{prod}$ that are not solely due to differences in null
rejections. For example, when $\beta_{2}^{(1)}=\beta_{2}^{(2)}=0.2$, $t_{IV}$
and $t_{prod}$ have similar null rejections with $b=0.5,1.0$. In these cases
power of $t_{IV}$ is higher for $\theta^{(2)}<1$, whereas power of $t_{prod}$
is higher for $\theta^{(2)}>1$. Power for both tests is higher with smaller
values of $b$ --- another well-known feature of tests based on kernel variance
estimators that use fixed-$b$ critical values. As trend slopes become smaller,
power of both tests decreases and power can be low even for values of
$\theta^{(2)}$ very far from $1$. Again, this is not surprising because
smaller trend slopes have relatively less information about $\theta^{(1)}$ and
$\theta^{(2)}$ for given strength of the noise. Interestingly, with very small
trend slopes, $\beta_{2}^{(1)}=\beta_{2}^{(2)}=0.05,0.025$, power initially
increases as $\theta^{(2)}$ moves away from the null value of $1$ but can
begin to fall when $\theta^{(2)}$ is very far from $1$.

The main takeaways from the power results are: i) for large trend slopes,
$t_{IV}$ and $t_{prod}$ have similar power, ii) for medium to small trend
slopes power cannot be ranked, iii) power of both tests decreases as the
bandwidth increases and iv) there is low power in detecting differences
between trend ratios when trend slopes are small.

\section{Practical Recommendations}

The theory and simulations indicate that $t_{IV}$ and $t_{prod}$ perform
similarly in practice when i) trend slopes are not very small, and ii) small
bandwidths are avoided when there is nontrivial serial correlation in the
noise. Both tests can over-reject when there is positive serial correlation in
the noise although this problem is mitigated by larger sample sizes. For very
small trend slopes both tests become conservative with $t_{IV}$ becoming
conservative more quickly than $t_{prod}$. Power of the tests cannot be ranked
in this case. Because of the conservative nature of both tests when trend
slopes are small (or even zero), any rejections obtained for tests of equal
ratios are robust to very small trend slopes. The price paid for this
robustness is lower power, but one cannot expect trend ratios to be precisely
estimated when trend slopes are very small.

If a practitioner uses either test (or both) with the Daniell kernel and a
non-small bandwidth for the variance estimator, then rejections of equal
ratios can be viewed as relatively robust to stationary serial correlation in
the noise and very small trend slopes.

\section{Empirical Application}

VMCS26 used the methods developed in this paper to test equivalence between
trend ratios of observed temperature series and temperature series generated
by recent runs of climate models. They compared trend ratios of temperature
series at various atmospheric heights relative to trends in surface
temperatures for the tropics region of the earth. These trend slope ratios are
called amplification ratios because climate models predict amplified warming
in the lower to mid troposphere relative to the surface in the tropics. VMCS26
found that climate models tend to have more amplification that seen in
observed temperature series and that the differences are statistically
significant using the $t_{IV}$ and $t_{prod}$ test statistics. While VMCS26
compared amplification ratios between each of five sets of observed
temperatures with each of 39 sets of climate model temperature series, they
did not compare amplification ratios among the five sets of observed
temperature series. Comparisons across observed series is interesting given
the different methods by which the observed series are measured, constructed,
and aggregated.

VMCS26 provide details on the five sets of observed temperature series which
can be summarized as follows. Observed temperatures for various pressure
levels in the atmosphere are measured in two ways in the five sets of observed
temperatures. The first method uses station-based, balloon-borne radiosonde
records. Temperature data is collected at specific pressure levels, generally
up to 20 hectopascals (hPa) which is approximately 27 km above the earth's
surface. A smaller hPa value indicates a higher level about the surface. The
second method is known as reanalyses of weather data combined with climate
models to generate temperature series. The five observed data series include
three from balloons. Two are data series from the University of Vienna,
RAOBCORE v1.9 and RICH v1.9. The third is from the U.S. National Oceanic and
Atmospheric Administration, the RATPAC-A v2 data series. These sets of data
are labeled RAOB, RICH, and RATP. The two global reanalyses data sets are
taken from the European Centre for Medium-Range Forecasts Reanalyses (ERA5),
and the Japanese Reanalyses for Three Quarters of a Century (JRA3Q).

The data is aggregated on an annual basis and spans the years 1958 to 2024 (67
years) for the surface and the grid of hPa levels 850, 700, 500, 400, 300,
200, 150, 100, 70, 50, 30, 20. Data is not available for the RATP data set for
hPa 20. Empirical results are presented in three tables where in all cases
confidence intervals at the 95\% level are computed using the same variance
estimators and fixed-$b$ critical values used in the finite sample simulations
(Daniell kernel, A91 bandwidth, fixed-$b$ critical values).

Table 3 provides summary statistics in the form of OLS estimated linear trend
slope parameters using regression (\ref{1.1}) for each of the five sets of
temperatures across the hPa levels. Estimated trend slopes and confidence
intervals are scaled to be in units of degrees Celsius per \textit{decade.}
Nearly all of the estimated trend slopes are statistically significantly
different from zero. For each of the five sets of series, there is a clear
pattern of trend slopes across hPa levels. There is warming at the surface.
There is less warming at 850 hPa but more warming from hPa levels 700 to 150.
At hPa 100 up to 20, there is cooling with cooling increasing at higher altitudes.

Table 4 reports estimated trend ratios for each hPa level relative to the
surface. Confidence intervals are computed using Fieller's method following
\cite{vogelsang-nawaz-JTSA}. All five sets of observed temperatures show the
same pattern in estimated trend ratios across hPa levels. Amplification
(ratios greater than 1) occurs for hPa levels 700 to 200. For hPa level 150
two observed series show amplification whereas three do not. For hPa levels
100 to 20, there is cooling at these higher altitudes and ratios are negative
and increase in magnitude for higher altitudes (lower hPa values). While the
patterns across hPa levels are similar for the five sets of observed
temperature series, there can be noticeable variation across the five sets for
a given hPa level. This raises the question as to whether ratios are equal
across pairs of observed series for a given hPa level.

Table 5 reports, for each hPa level, pairwise differences between estimated
ratios, $\widehat{\Delta}_{\theta}$, and pairwise $g_{\widehat{\beta}}$
values. 95\% fixed-$b$ critical values are given below each estimate. A
$^{\ast}$ superscript on an estimated value indicates a rejection of the null
hypothesis of equal ratios (indicates the confidence interval does not contain
the value $0$). For reporting purposes, the values of $g_{\widehat{\beta}}$
and confidence intervals are scaled by $10^{4}$. Of course, this has no effect
on whether the null hypothesis of equal ratios can be rejected. Table 5 is
divided into four panels where each panel reports results for three hPa
levels. In general, there are many pairs where the differences in trend ratios
are not small in magnitude and are statistically significant. For example,
among the three balloon data sets, RICH, RAOB, RATP, there are rejections of
equal ratios for 8 to 11 of the hPa levels across the three respective pairs.
One case where there is relative alignment of trend ratios is between the two
Reanalyses data sets, ERA5 and JRA3Q, where rejections of equal ratios occurs
for only 3 of 12 hPa levels. Except between the pair of Reanalyses data sets,
the results in Table 5 suggest there are differences in trend ratios among
pairs of the observed data sets across hPa levels.

\section{Conclusion}

This paper develops estimation and inference methods for systems of pairs of
trend stationary time series where the parameters of interest are ratios of
trend slopes between pairs of time series. Inference focuses on null
hypotheses that can be written as linear restrictions across trend ratios.
Trend slopes are estimated by IV using time as an instrument and test
statistics are robust to i) stationary serial correlation in the fluctuations
around trend, and ii) correlation between and across pairs of time series. For
the empirically relevant special case of testing for equal trend slopes
between two pairs of time series, an alternative testing approach is developed
by restating the equal trend ratio restriction as a restriction involving
products of the underlying trend slopes which are estimated by OLS. Theory and
finite sample results suggest that both the IV and products approach work well
and have similar properties when trend slopes are not too small and serial
correlation is not too strong. When trend slopes are very small relative to
the variation around the trend function, the null limiting distribution of
both tests become nonpivotal and both tests tend to under-reject under the
null and have low power. Lower power is expected when trend slopes are very
small because trend ratios cannot be precisely estimated. That both tests
become conservative when trend slopes are very small gives the tests useful
robustness to very small trend slopes. Overall, the IV and product approaches
have complementary finite sample properties and using both is recommended in
practice for testing the hypothesis of equal ratios across two pairs of time series.

\section{Appendix: Proofs of Theorems}

The following lemmas provide the limits of various terms that appear in the IV
estimator of $\theta$, the OLS estimators of the trend slopes, and the
variance estimators. The theorems are straightforward to establish using
algebra, the continuous mapping theorem (CMT), and the lemmas. The first lemma
gives results that hold for any magnitude of trend slopes. Subsequent lemmas
are given for the three cases of trend slope magnitudes.

\begin{lemma}
Suppose that (\ref{fclt}) and (\ref{fclteps}) hold. The following hold as
$T\rightarrow\infty$ for any values of $\beta_{1}^{(i)},\beta_{2}^{(i)}:$%
\[
T^{-3}%
{\displaystyle\sum\limits_{t=1}^{T}}
(t-\overline{t})^{2}\overset{}{\rightarrow}\int_{0}^{1}(s-\frac{1}{2}%
)^{2}ds=\frac{1}{12},
\]%
\[
T^{-2}%
{\displaystyle\sum\limits_{t=1}^{[rT]}}
(t-\overline{t})\overset{}{\rightarrow}\int_{0}^{r}(s-\frac{1}{2})ds=L(r),
\]%
\[
T^{-1/2}%
{\displaystyle\sum\limits_{t=1}^{[rT]}}
\left(  \scalebox{1.5}{$\bm{\epsilon}$}_{\theta t}-\overline
{\scalebox{1.5}{$\bm{\epsilon}$}}_{\theta}\right)  \overset{}{\Rightarrow
}\mathbf{\Lambda}_{\epsilon}\left(  \mathbf{W}_{\epsilon}(r)-r\mathbf{W}%
_{\epsilon}(1)\right)  \equiv\mathbf{\Lambda}_{\epsilon}\widetilde{\mathbf{W}%
}_{\epsilon}(r),
\]%
\[
T^{-3/2}%
{\displaystyle\sum\limits_{t=1}^{T}}
(t-\overline{t})\scalebox{1.5}{$\bm{\epsilon}$}_{\theta t}%
\overset{}{\Rightarrow}\mathbf{\Lambda}_{\epsilon}\int_{0}^{1}(s-\frac{1}%
{2})d\mathbf{W}_{\epsilon}(s),
\]%
\[
T^{-3/2}%
{\displaystyle\sum\limits_{t=1}^{T}}
(t-\overline{t})\mathbf{U}_{t}\overset{}{\Rightarrow}\int_{0}^{1}(s-\frac
{1}{2})d\mathbf{B}_{u}(s)=\mathbf{\Lambda}_{u}\int_{0}^{1}(s-\frac{1}%
{2})d\mathbf{W}_{u}(s),
\]%
\[
T^{-3/2}\left(  \widehat{\mathbf{\beta}}-\mathbf{\beta}\right)
\overset{}{\Rightarrow}12\mathbf{\Lambda}_{u}\int_{0}^{1}(s-\frac{1}%
{2})d\mathbf{W}_{u}(s)=12\int_{0}^{1}(s-\frac{1}{2})d\mathbf{B}_{u}%
(s)=\mathbf{\Psi},
\]%
\begin{align*}
T^{-1/2}%
{\displaystyle\sum\limits_{t=1}^{[rT]}}
\widehat{\mathbf{U}}_{t}  &  \Rightarrow\left[  \mathbf{B}_{u}(r)-r\mathbf{B}%
_{u}(1)-12L(r)\int_{0}^{1}(s-\frac{1}{2})d\mathbf{B}_{u}(s)\right]
\equiv\widetilde{\mathbf{B}}_{u}(r)\\
&  =\mathbf{\Lambda}_{u}\widetilde{\mathbf{W}}_{u}(r)\equiv\mathbf{\Lambda
}_{u}\left[  \mathbf{W}_{u}(r)-r\mathbf{W}_{u}(1)-12L(r)\int_{0}^{1}%
(s-\frac{1}{2})d\mathbf{W}_{u}(s)\right]  .
\end{align*}

\end{lemma}

\noindent\textbf{Proof: }The results in this lemma are standard given the
FCLTs (\ref{fclt}) and (\ref{fclteps}). See \cite{hamilton-book}.

\begin{lemma}
(Large to small slopes) Suppose that (\ref{fclt}) and (\ref{fclteps}) hold and
$\beta_{1}^{(i)}=T^{-\kappa}\overline{\beta}_{1}^{(i)}$, $\beta_{2}%
^{(i)}=T^{-\kappa}\overline{\beta}_{2}^{(i)}$ with $0\leq\kappa<\frac{3}{2}$.
The following hold as $T\rightarrow\infty$,%
\begin{gather*}
T^{-3+\kappa}%
{\displaystyle\sum\limits_{t=1}^{T}}
(t-\overline{t})(y_{2t}^{(i)}-\overline{y}_{2}^{(i)})\overset{p}{\rightarrow
}\overline{\beta}_{2}^{(i)}\int_{0}^{1}(s-\frac{1}{2})^{2}ds=\frac{1}%
{12}\overline{\beta}_{2}^{(i)},\\
T^{-2+\kappa}%
{\displaystyle\sum\limits_{t=1}^{[rT]}}
(y_{2t}^{(i)}-\overline{y}_{2}^{(i)})\overset{p}{\rightarrow}\overline{\beta
}_{2}^{(i)}L(r),\\
T^{\kappa}\mathbf{R}_{\widehat{\beta}}\overset{p}{\rightarrow}\mathbf{R}%
_{\overline{\beta}}.
\end{gather*}

\end{lemma}

\noindent\textbf{Proof: }The first two results of the lemma are easy to
establish once%
\[
y_{2t}^{(i)}-\overline{y}_{2}^{(i)}=\beta_{2}^{(i)}\left(  t-\overline
{t}\right)  +(u_{2t}^{(i)}-\overline{u}_{2}^{(i)}),
\]
is substituted into each expression and applying limits from Lemma 1:%
\begin{align*}
T^{-3+\kappa}%
{\displaystyle\sum\limits_{t=1}^{T}}
(t-\overline{t})(y_{2t}^{(i)}-\overline{y}_{2}^{(i)})  &  =T^{-3+\kappa}%
\beta_{2}^{(i)}%
{\displaystyle\sum\limits_{t=1}^{T}}
(t-\overline{t})^{2}+T^{-3/2+\kappa}T^{-3/2}%
{\displaystyle\sum\limits_{t=1}^{T}}
(t-\overline{t})u_{2t}^{(i)}\\
&  =\overline{\beta}_{2}^{(i)}T^{-3}%
{\displaystyle\sum\limits_{t=1}^{T}}
(t-\overline{t})^{2}+o_{p}(1)\overset{p}{\rightarrow}\overline{\beta}%
_{2}^{(i)}\int_{0}^{1}(s-\frac{1}{2})^{2}ds=\frac{1}{12}\overline{\beta}%
_{2}^{(i)},
\end{align*}%
\begin{align*}
T^{-2+\kappa}%
{\displaystyle\sum\limits_{t=1}^{[rT]}}
(y_{2t}^{(i)}-\overline{y}_{2}^{(i)})  &  =T^{-2+\kappa}\beta_{2}^{(i)}%
{\displaystyle\sum\limits_{t=1}^{[rT]}}
\left(  t-\overline{t}\right)  +T^{-3/2+\kappa}T^{-1/2}%
{\displaystyle\sum\limits_{t=1}^{[rT]}}
(u_{2t}^{(i)}-\overline{u}_{2}^{(i)})\\
&  =\overline{\beta}_{2}^{(i)}T^{-2}%
{\displaystyle\sum\limits_{t=1}^{[rT]}}
\left(  t-\overline{t}\right)  +o_{p}(1)\overset{p}{\rightarrow}%
\overline{\beta}_{2}^{(i)}L(r).
\end{align*}
The second two terms of each expression are $o_{p}(1)$ because $T^{-3/2+\kappa
}\rightarrow0$ as $T\rightarrow\infty$ for $0\leq\kappa<\frac{3}{2}$. For the
third result of the lemma note that%
\[
\widehat{\beta}_{j}^{(i)}=\beta_{j}^{(i)}+\left(  \widehat{\beta}_{j}%
^{(i)}-\beta_{j}^{(i)}\right)  =T^{-\kappa}\overline{\beta}_{j}^{(i)}+\left(
\widehat{\beta}_{j}^{(i)}-\beta_{j}^{(i)}\right)  ,
\]
and it follows that%
\begin{equation}
T^{\kappa}\widehat{\beta}_{j}^{(i)}=\overline{\beta}_{j}^{(i)}+T^{\kappa
}\left(  \widehat{\beta}_{j}^{(i)}-\beta_{j}^{(i)}\right)  =\overline{\beta
}_{j}^{(i)}+T^{-3/2+\kappa}T^{3/2}\left(  \widehat{\beta}_{j}^{(i)}-\beta
_{j}^{(i)}\right)  \overline{\beta}_{j}^{(i)}=\overline{\beta}_{j}^{(i)}%
+o_{p}(1), \label{lem2_1}%
\end{equation}
where second term is $o_{p}(1)$ because $0\leq\kappa<\frac{3}{2}$. Using
(\ref{lem2_1}) it easily follows that%
\[
T^{\kappa}\mathbf{R}_{\widehat{\beta}}=\left[  T^{\kappa}\widehat{\beta}%
_{2}^{(2)},-T^{\kappa}\widehat{\beta}_{2}^{(1)},-T^{\kappa}\widehat{\beta}%
_{1}^{(2)},T^{\kappa}\widehat{\beta}_{1}^{(1)}\right]  =\left[  \overline
{\beta}_{2}^{(2)},-\overline{\beta}_{2}^{(1)},-\overline{\beta}_{1}%
^{(2)},\overline{\beta}_{1}^{(1)}\right]  +o_{p}(1)\overset{p}{\rightarrow
}\mathbf{R}_{\overline{\beta}}.
\]

\bigskip

\begin{lemma}
(Very small slopes) Suppose that (\ref{fclt}) and (\ref{fclteps}) hold and
$\beta_{1}^{(i)}=T^{-3/2}\overline{\beta}_{1}^{(i)}$, $\beta_{2}%
^{(i)}=T^{-3/2}\overline{\beta}_{2}^{(i)}$ $(\kappa=\frac{3}{2})$. The
following hold as $T\rightarrow\infty$:%
\begin{gather*}
T^{-3/2}%
{\displaystyle\sum\limits_{t=1}^{T}}
(t-\overline{t})(y_{2t}^{(i)}-\overline{y}_{2}^{(i)})\Rightarrow\frac{1}%
{12}\overline{\beta}_{2}^{(i)}+\int_{0}^{1}(s-\frac{1}{2})dB_{u2}^{(i)}(s),\\
T^{-1/2}%
{\displaystyle\sum\limits_{t=1}^{[rT]}}
(y_{2t}^{(i)}-\overline{y}_{2}^{(i)})\Rightarrow\overline{\beta}_{2}%
^{(i)}L(r)+B_{u2}^{(i)}(r)-rB_{u2}^{(i)}(1)=\overline{\beta}_{2}%
^{(i)}L(r)+\widetilde{B}_{u2}^{(i)}(r),\\
T^{3/2}\mathbf{R}_{\widehat{\beta}}\Rightarrow\mathbf{R}_{\overline{\beta}%
}+\left[  \Psi_{2}^{(2)},-\Psi_{2}^{(1)},-\Psi_{1}^{(2)},\Psi_{1}%
^{(1)}\right]  .
\end{gather*}

\end{lemma}

\noindent\textbf{Proof: }Setting $\kappa=\frac{3}{2}$ in the proof of Lemma 2
and using Lemma 1 gives%
\[
T^{-3/2}%
{\displaystyle\sum\limits_{t=1}^{T}}
(t-\overline{t})(y_{2t}^{(i)}-\overline{y}_{2}^{(i)})=\overline{\beta}%
_{2}^{(i)}T^{-3}%
{\displaystyle\sum\limits_{t=1}^{T}}
(t-\overline{t})^{2}+T^{-3/2}%
{\displaystyle\sum\limits_{t=1}^{T}}
(t-\overline{t})u_{2t}^{(i)}\Rightarrow\frac{1}{12}\overline{\beta}_{2}%
^{(i)}+\int_{0}^{1}(s-\frac{1}{2})dB_{u2}^{(i)}(s),
\]%
\[
T^{-1/2}%
{\displaystyle\sum\limits_{t=1}^{[rT]}}
(y_{2t}^{(i)}-\overline{y}_{2}^{(i)})=\overline{\beta}_{2}^{(i)}T^{-2}%
{\displaystyle\sum\limits_{t=1}^{[rT]}}
\left(  t-\overline{t}\right)  +T^{-1/2}%
{\displaystyle\sum\limits_{t=1}^{[rT]}}
(u_{2t}^{(i)}-\overline{u}_{2}^{(i)})\Rightarrow\overline{\beta}_{2}%
^{(i)}L(r)+\widetilde{B}_{u2}^{(i)}(r),
\]%
\begin{equation}
T^{3/2}\widehat{\beta}_{j}^{(i)}=\overline{\beta}_{j}^{(i)}+T^{3/2}\left(
\widehat{\beta}_{j}^{(i)}-\beta_{j}^{(i)}\right)  \Rightarrow\overline{\beta
}_{j}^{(i)}+\Psi_{j}^{(i)}, \label{lem3_1}%
\end{equation}
where the second part of (\ref{lem3_1}) follows from elements of the limit of
the vector $T^{-3/2}\left(  \widehat{\mathbf{\beta}}-\mathbf{\beta}\right)  $
in Lemma 1. The limit of $T^{3/2}\mathbf{R}_{\widehat{\beta}}$ directly
follows from (\ref{lem3_1}).

\bigskip

\begin{lemma}
(Zero slopes) Suppose that (\ref{fclt}) and (\ref{fclteps}) hold and
$\beta_{1}^{(i)}=0$, $\beta_{2}^{(i)}=0$. The following hold as $T\rightarrow
\infty$,%
\begin{gather*}
T^{-3/2}%
{\displaystyle\sum\limits_{t=1}^{T}}
(t-\overline{t})(y_{2t}^{(i)}-\overline{y}_{2}^{(i)})\Rightarrow\int_{0}%
^{1}(s-\frac{1}{2})dB_{u2}^{(i)}(s),\\
T^{-1/2}%
{\displaystyle\sum\limits_{t=1}^{[rT]}}
(y_{2t}^{(i)}-\overline{y}_{2}^{(i)})\Rightarrow B_{u2}^{(i)}(r)-rB_{u2}%
^{(i)}(1)=\widetilde{B}_{u2}^{(i)}(r),\\
T^{3/2}\mathbf{R}_{\widehat{\beta}}\Rightarrow\left[  \Psi_{2}^{(2)},-\Psi
_{2}^{(1)},-\Psi_{1}^{(2)},\Psi_{1}^{(1)}\right]  .
\end{gather*}

\end{lemma}

\noindent\textbf{Proof: }The results follow directly from Lemma 3 by setting
$\overline{\beta}_{j}^{(i)}=0$ in all the limiting expressions.

\bigskip

\bigskip

\noindent\textbf{Proof of Theorem 1.} First consider the cases where the trend
slopes are not zero. Using (\ref{thetahat_centered}) and scaling by
$T^{3/2-\kappa}$ gives%
\[
T^{3/2-\kappa}\left(  \widehat{\mathbf{\theta}}-\mathbf{\theta}\right)
\mathbf{=}\left(  T^{-3+\kappa}\widehat{\mathbf{D}}_{2}\right)  ^{-1}%
T^{-3/2}\sum_{t=1}^{T}\left(  t-\overline{t}\right)
\scalebox{1.5}{$\bm{\epsilon}$}_{\theta t}.
\]
For the case of large to small slopes, the limit of the $i^{th}$ diagonal
element of $T^{-3+\kappa}\widehat{\mathbf{D}}_{2}$, which is given by
$T^{-3+\kappa}%
{\displaystyle\sum_{t=1}^{T}}
(t-\overline{t})(y_{2t}^{(i)}-\overline{y}_{2}^{(i)})$, follows from Lemma 2,
and it follows that%
\[
T^{-3+\kappa}\widehat{\mathbf{D}}_{2}\overset{p}{\rightarrow}\frac{1}%
{12}\mathbf{D}_{\overline{\beta}_{2}}.
\]
Using the limit of $T^{-3/2}\sum_{t=1}^{T}\left(  t-\overline{t}\right)
\scalebox{1.5}{$\bm{\epsilon}$}_{\theta t}$ from Lemma 1 gives%
\[
T^{3/2-\kappa}\left(  \widehat{\mathbf{\theta}}-\mathbf{\theta}\right)
\Rightarrow\left(  \frac{1}{12}\mathbf{D}_{\overline{\beta}_{2}}\right)
^{-1}\mathbf{\Lambda}_{\epsilon}\int_{0}^{1}(s-\frac{1}{2})d\mathbf{W}%
_{\epsilon}(s),
\]
as required. For the case of very small slopes, Lemma 3 gives%
\[
T^{-3/2}\widehat{\mathbf{D}}_{2}\overset{p}{\rightarrow}\frac{1}{12}%
\mathbf{D}_{\overline{\beta}_{2}}+\mathbf{D}_{\mathbf{B}_{u2}},
\]
and it follows that%
\[
\left(  \widehat{\mathbf{\theta}}-\mathbf{\theta}\right)  \Rightarrow\left(
\frac{1}{12}\mathbf{D}_{\overline{\beta}_{2}}+\mathbf{D}_{\mathbf{B}_{u2}%
}\right)  ^{-1}\mathbf{\Lambda}_{\epsilon}\int_{0}^{1}(s-\frac{1}%
{2})d\mathbf{W}_{\epsilon}(s),
\]
as required. For the case of zero slopes, from (\ref{thetahat_zero_slopes}) it
follows that%
\[
\widehat{\theta}^{(i)}=\frac{%
{\textstyle\sum\nolimits_{t=1}^{T}}
(t-\overline{t})u_{1t}^{(i)}}{%
{\textstyle\sum\nolimits_{t=1}^{T}}
(t-\overline{t})u_{2t}^{(i)}}=\frac{T^{-3/2}%
{\textstyle\sum\nolimits_{t=1}^{T}}
(t-\overline{t})u_{1t}^{(i)}}{T^{-3/2}%
{\textstyle\sum\nolimits_{t=1}^{T}}
(t-\overline{t})u_{2t}^{(i)}}\Rightarrow\frac{\int_{0}^{1}(s-\frac{1}%
{2})dB_{u1}^{(i)}(s)}{\int_{0}^{1}(s-\frac{1}{2})dB_{u2}^{(i)}(s)},
\]
where convergence follows using elements from the limit of $T^{-3/2}%
{\displaystyle\sum_{t=1}^{T}}
(t-\overline{t})\mathbf{U}_{t}$ from Lemma 1. Collecting the $\widehat{\theta
}^{(i)}$ into the vector $\widehat{\mathbf{\theta}}$ gives the result.

\bigskip

\noindent\textbf{Proof of Theorem 2.} Results are given for the Wald
statistic. The limits of the $t$-statistics following easily from the Wald
statistic arguments. The Wald statistic is given by%
\[
Wald_{IV}=\left(  \mathbf{R}\widehat{\mathbf{\theta}}-\mathbf{r}\right)
^{\prime}\left[  \mathbf{R}\left(
{\textstyle\sum\nolimits_{t=1}^{T}}
(t-\overline{t})^{2}\right)  \widehat{\mathbf{D}}_{2}^{-1}%
\widehat{\mathbf{\Omega}}_{\mathbf{\epsilon}}\widehat{\mathbf{D}}_{2}%
^{-1}\mathbf{R}^{\prime}\right]  ^{-1}\left(  \mathbf{R}%
\widehat{\mathbf{\theta}}-\mathbf{r}\right)  .
\]
First consider the cases where the trend slopes are nonzero. Rewrite
$\mathbf{R}\widehat{\mathbf{\theta}}-\mathbf{r}$ as%
\[
\mathbf{R}\widehat{\mathbf{\theta}}-\mathbf{r=R}\left(
\widehat{\mathbf{\theta}}-\mathbf{\theta}\right)  +\mathbf{R\theta
}-\mathbf{r.}%
\]
Under $H_{1L}$ it follows that $\mathbf{R\theta=r+}\overline{\mathbf{\Delta}%
}T^{-3/2+\kappa}$, or equivalently $\mathbf{R\theta-r=}\overline
{\mathbf{\Delta}}T^{-3/2+\kappa}$ giving%
\begin{equation}
\mathbf{R}\widehat{\mathbf{\theta}}-\mathbf{r=R}\left(
\widehat{\mathbf{\theta}}-\mathbf{\theta}\right)  +\overline{\mathbf{\Delta}%
}T^{-3/2+\kappa}. \label{rthetahat2}%
\end{equation}
Plugging (\ref{rthetahat2}) into $Wald_{IV}$ gives%
\[
Wald_{IV}=\left(  \mathbf{R}\left(  \widehat{\mathbf{\theta}}-\mathbf{\theta
}\right)  +\overline{\mathbf{\Delta}}T^{-3/2+\kappa}\right)  ^{\prime}\left[
\mathbf{R}\left(
{\textstyle\sum\nolimits_{t=1}^{T}}
(t-\overline{t})^{2}\right)  \widehat{\mathbf{D}}_{2}^{-1}%
\widehat{\mathbf{\Omega}}_{\mathbf{\epsilon}}\widehat{\mathbf{D}}_{2}%
^{-1}\mathbf{R}^{\prime}\right]  ^{-1}\left(  \mathbf{R}\left(
\widehat{\mathbf{\theta}}-\mathbf{\theta}\right)  +\overline{\mathbf{\Delta}%
}T^{-3/2+\kappa}\right)
\]%
\begin{align*}
&  =\left(  \mathbf{R}T^{3/2-\kappa}\left(  \widehat{\mathbf{\theta}%
}-\mathbf{\theta}\right)  +\overline{\mathbf{\Delta}}\right)  ^{\prime}\left[
\mathbf{R}\left(  T^{-3}%
{\textstyle\sum\nolimits_{t=1}^{T}}
(t-\overline{t})^{2}\right)  \left(  T^{-3+\kappa}\widehat{\mathbf{D}}%
_{2}\right)  ^{-1}\widehat{\mathbf{\Omega}}_{\mathbf{\epsilon}}\left(
T^{-3+\kappa}\widehat{\mathbf{D}}_{2}\right)  ^{-1}\mathbf{R}^{\prime}\right]
^{-1}\\
&  \text{ \ \ \ \ \ \ \ \ \ \ \ \ \ \ \ \ \ \ \ \ \ \ \ \ \ \ }\times\left(
\mathbf{R}T^{3/2-\kappa}\left(  \widehat{\mathbf{\theta}}-\mathbf{\theta
}\right)  +\overline{\mathbf{\Delta}}\right)  .
\end{align*}
With the exception of $\widehat{\mathbf{\Omega}}_{\mathbf{\epsilon}}$, limits
of the scaled components of $Wald_{IV}$ follow from the proof of Theorem 1.
Following \cite{kvnewasym} the limit of $\widehat{\mathbf{\Omega}%
}_{\mathbf{\epsilon}}$ under the fixed-$b$ nesting for the bandwidth is given
by $\mathbf{P}_{b}(\mathbf{Q}(r))$ where $\mathbf{Q}(r)$ is the corresponding
limit of $T^{-1/2}\sum_{t=1}^{[rT]}\widehat{\scalebox{1.5}{$\bm{\epsilon}$}}%
_{\theta t}$. The form of $\mathbf{Q}(r)$ depends on whether the trend slopes
are large to small or are very small.

Simple algebra gives%
\[
T^{-1/2}\sum_{t=1}^{[rT]}\widehat{\scalebox{1.5}{$\bm{\epsilon}$}}_{\theta
t}=T^{-1/2}%
{\displaystyle\sum\limits_{t=1}^{[rT]}}
\left(  \scalebox{1.5}{$\bm{\epsilon}$}_{\theta t}-\overline
{\scalebox{1.5}{$\bm{\epsilon}$}}_{\theta}\right)  -\left(  T^{-2+\kappa}%
{\displaystyle\sum\limits_{t=1}^{[rT]}}
\mathbf{H}_{t}\right)  T^{3/2-\kappa}\left(  \widehat{\mathbf{\theta}%
}-\mathbf{\theta}\right)  ,
\]
where $\mathbf{H}_{t}$ is an $n\times n$ diagonal matrix with $i^{th}$
diagonal element $\sum_{t=1}^{[rT]}(y_{2t}^{(i)}-\overline{y}_{2}^{(i)})$. For
the case of large to small trend slopes, it follows from Lemma 2 that%
\begin{equation}
T^{-2+\kappa}%
{\displaystyle\sum\limits_{t=1}^{[rT]}}
\mathbf{H}_{t}\overset{p}{\rightarrow}\mathbf{D}_{\overline{\beta}_{2}}L(r).
\label{Hlim1}%
\end{equation}
Using (\ref{Hlim1}), Lemma 1, Lemma 2, and Theorem 1 it follows that%
\begin{align*}
&  T^{-1/2}\sum_{t=1}^{[rT]}\widehat{\scalebox{1.5}{$\bm{\epsilon}$}}_{\theta
t}\overset{}{\Rightarrow}\mathbf{\Lambda}_{\epsilon}\left(  \mathbf{W}%
_{\epsilon}(r)-r\mathbf{W}_{\epsilon}(1)\right)  -\mathbf{D}_{\overline{\beta
}_{2}}L(r)\left(  \frac{1}{12}\mathbf{D}_{\overline{\beta}_{2}}\right)
^{-1}\mathbf{\Lambda}_{\epsilon}\int_{0}^{1}(s-\frac{1}{2})d\mathbf{W}%
_{\epsilon}(s)\\
&  =\mathbf{\Lambda}_{\epsilon}\left(  \mathbf{W}_{\epsilon}(r)-r\mathbf{W}%
_{\epsilon}(1)\right)  -12L(r)\mathbf{\Lambda}_{\epsilon}\int_{0}^{1}%
(s-\frac{1}{2})d\mathbf{W}_{\epsilon}(s)\\
&  =\mathbf{\Lambda}_{\epsilon}\left(  \mathbf{W}_{\epsilon}(r)-r\mathbf{W}%
_{\epsilon}(1)-12L(r)\int_{0}^{1}(s-\frac{1}{2})d\mathbf{W}_{\epsilon
}(s)\right)  \equiv\mathbf{\Lambda}_{\epsilon}\widetilde{\mathbf{W}%
}_{\mathbf{\epsilon}}(r).
\end{align*}
For the case of very small slopes, $\kappa=\frac{3}{2}$, it follow from Lemma
3 that%
\begin{equation}
T^{-1/2}%
{\displaystyle\sum\limits_{t=1}^{[rT]}}
\mathbf{H}_{t}\overset{p}{\rightarrow}\mathbf{D}_{\overline{\beta}_{2}%
}L(r)+\mathbf{D}_{\widehat{\mathbf{B}}_{u2}(r)}. \label{Hlim2}%
\end{equation}
Using (\ref{Hlim2}), Lemma 1, Lemma 3, and Theorem 1 it follows that%
\[
T^{-1/2}\sum_{t=1}^{[rT]}\widehat{\scalebox{1.5}{$\bm{\epsilon}$}}_{\theta
t}\overset{}{\Rightarrow}\mathbf{H}_{1}(r),
\]%
\[
\mathbf{H}_{1}(r)=\mathbf{\Lambda}_{\epsilon}\left(  \mathbf{W}_{\epsilon
}(r)-r\mathbf{W}_{\epsilon}(1)\right)  -\left(  \mathbf{D}_{\overline{\beta
}_{2}}L(r)+\mathbf{D}_{\widehat{\mathbf{B}}_{u2}(r)}\right)  \left(  \frac
{1}{12}\mathbf{D}_{\overline{\beta}_{2}}+\mathbf{D}_{\mathbf{B}_{u2}}\right)
^{-1}\mathbf{\Lambda}_{\epsilon}\int_{0}^{1}(s-\frac{1}{2})d\mathbf{W}%
_{\epsilon}(s).
\]
The limit of $Wald_{IV}$ is now straightforward to obtain. For the case of
large to small trend slopes it follows that%
\begin{align*}
Wald_{IV}  &  =\left(  \mathbf{R}T^{3/2-\kappa}\left(  \widehat{\mathbf{\theta
}}-\mathbf{\theta}\right)  +\overline{\mathbf{\Delta}}\right)  ^{\prime
}\left[  \mathbf{R}\left(  T^{-3}%
{\displaystyle\sum\limits_{t=1}^{T}}
(t-\overline{t})^{2}\right)  \left(  T^{-3+\kappa}\widehat{\mathbf{D}}%
_{2}\right)  ^{-1}\widehat{\mathbf{\Omega}}_{\mathbf{\epsilon}}\left(
T^{-3+\kappa}\widehat{\mathbf{D}}_{2}\right)  ^{-1}\mathbf{R}^{\prime}\right]
^{-1}\\
&  \text{ \ \ \ \ \ \ \ \ \ \ \ \ \ \ \ \ \ \ \ \ \ \ \ \ \ \ \ }\times\left(
\mathbf{R}T^{3/2-\kappa}\left(  \widehat{\mathbf{\theta}}-\mathbf{\theta
}\right)  +\overline{\mathbf{\Delta}}\right)  ,
\end{align*}%
\begin{align*}
&  \Rightarrow\left(  \mathbf{R}\left(  \frac{1}{12}\mathbf{D}_{\overline
{\beta}_{2}}\right)  ^{-1}\mathbf{\Lambda}_{\epsilon}\int_{0}^{1}(s-\frac
{1}{2})d\mathbf{W}_{\epsilon}(s)+\overline{\mathbf{\Delta}}\right)  ^{\prime
}\left[  \mathbf{R}\frac{1}{12}\left(  \frac{1}{12}\mathbf{D}_{\overline
{\beta}_{2}}\right)  ^{-1}\mathbf{\Lambda}_{\epsilon}\mathbf{P}_{b}%
(\widetilde{\mathbf{W}}_{\epsilon}(r))\mathbf{\Lambda}_{\epsilon}^{\prime
}\left(  \frac{1}{12}\mathbf{D}_{\overline{\beta}_{2}}\right)  ^{-1}%
\mathbf{R}^{\prime}\right]  ^{-1}\\
&  \text{ \ \ \ \ \ \ \ \ \ \ \ \ \ \ \ \ \ \ \ \ \ \ \ \ \ \ \ \ \ \ \ }%
\times\left(  \mathbf{R}\left(  \frac{1}{12}\mathbf{D}_{\overline{\beta}_{2}%
}\right)  ^{-1}\mathbf{\Lambda}_{\epsilon}\int_{0}^{1}(s-\frac{1}%
{2})d\mathbf{W}_{\epsilon}(s)+\overline{\mathbf{\Delta}}\right)  ,
\end{align*}%
\begin{align*}
&  =\left(  \mathbf{12RD}_{\overline{\beta}_{2}}^{-1}\mathbf{\Lambda
}_{\epsilon}\int_{0}^{1}(s-\frac{1}{2})d\mathbf{W}_{\epsilon}(s)+\overline
{\mathbf{\Delta}}\right)  ^{\prime}\left[  12\mathbf{RD}_{\overline{\beta}%
_{2}}^{-1}\mathbf{\Lambda}_{\epsilon}\mathbf{P}_{b}(\widetilde{\mathbf{W}%
}_{\epsilon}(r))\mathbf{\Lambda}_{\epsilon}^{\prime}\mathbf{D}_{\overline
{\beta}_{2}}^{-1}\mathbf{R}^{\prime}\right]  ^{-1}\\
&  \text{ \ \ \ \ \ \ \ \ \ \ \ \ \ \ \ \ \ \ \ \ \ \ \ \ \ \ \ \ }%
\times\left(  \mathbf{12RD}_{\overline{\beta}_{2}}^{-1}\mathbf{\Lambda
}_{\epsilon}\int_{0}^{1}(s-\frac{1}{2})d\mathbf{W}_{\epsilon}(s)+\overline
{\mathbf{\Delta}}\right)  .
\end{align*}
The form of the limit given in Theorem 2 is obtained by replacing
$\mathbf{RD}_{\overline{\beta}_{2}}^{-1}\mathbf{\Lambda}_{\epsilon}%
\mathbf{W}_{\epsilon}(s)$ with $\mathbf{\Lambda}_{\epsilon}^{\ast}%
\mathbf{W}_{\epsilon}^{\ast}(s)$ where $\mathbf{W}_{\epsilon}^{\ast}(s)$ is a
$q\times1$ vector of standard Wiener processes and $\mathbf{\Lambda}%
_{\epsilon}^{\ast}$ is the matrix square root of the $q\times q$ matrix
$\mathbf{RD}_{\overline{\beta}_{2}}^{-1}\mathbf{\Lambda}_{\epsilon
}\mathbf{\Lambda}_{\epsilon}^{\prime}\mathbf{D}_{\overline{\beta}_{2}}%
^{-1}\mathbf{R}^{\prime}$. The $\mathbf{RD}_{\overline{\beta}_{2}}%
^{-1}\mathbf{\Lambda}_{\epsilon}$ matrices on both sides of $\mathbf{P}%
_{b}(\widetilde{\mathbf{W}}_{\epsilon}(r))$ can be pushed inside
$\mathbf{P}_{b}(\mathbf{\cdot})$ to give $\mathbf{RD}_{\overline{\beta}_{2}%
}^{-1}\mathbf{\Lambda}_{\epsilon}\mathbf{P}_{b}(\widetilde{\mathbf{W}%
}_{\epsilon}(r))\mathbf{\Lambda}_{\epsilon}^{\prime}\mathbf{D}_{\overline
{\beta}_{2}}^{-1}\mathbf{R}^{\prime}=\mathbf{P}_{b}(\mathbf{RD}_{\overline
{\beta}_{2}}^{-1}\mathbf{\Lambda}_{\epsilon}\widetilde{\mathbf{W}}_{\epsilon
}(r))\equiv\mathbf{P}_{b}(\mathbf{\Lambda}_{\epsilon}^{\ast}%
\widetilde{\mathbf{W}}_{\epsilon}^{\ast}(r))=\mathbf{\Lambda}_{\epsilon}%
^{\ast}\mathbf{P}_{b}(\widetilde{\mathbf{W}}_{\epsilon}^{\ast}%
(r))\mathbf{\Lambda}_{\epsilon}^{\ast\prime}$ where $\widetilde{\mathbf{W}%
}_{\epsilon}^{\ast}(r)$ has the same form as $\widetilde{\mathbf{W}}%
_{\epsilon}(r)$ with $\mathbf{W}_{\epsilon}^{\ast}(s)$ in place of
$\mathbf{W}_{\epsilon}(s)$. Therefore, the limit of $Wald_{IV}$ is given by%
\[
Wald_{IV}\Rightarrow\left(  \mathbf{12\Lambda}_{\epsilon}^{\ast}\int_{0}%
^{1}(s-\frac{1}{2})d\mathbf{W}_{\epsilon}(s)+\overline{\mathbf{\Delta}%
}\right)  ^{\prime}\left[  12\mathbf{\Lambda}_{\epsilon}^{\ast}\mathbf{P}%
_{b}(\widetilde{\mathbf{W}}_{\epsilon}^{\ast}(r))\mathbf{\Lambda}_{\epsilon
}^{\ast\prime}\right]  ^{-1}\left(  \mathbf{12\Lambda}_{\epsilon}^{\ast}%
\int_{0}^{1}(s-\frac{1}{2})d\mathbf{W}_{\epsilon}(s)+\overline{\mathbf{\Delta
}}\right)  ,
\]%
\[
=\left(  \sqrt{\mathbf{12}}\int_{0}^{1}(s-\frac{1}{2})d\mathbf{W}_{\epsilon
}(s)+\frac{1}{\sqrt{\mathbf{12}}}\mathbf{\Lambda}_{\epsilon}^{\ast-1}%
\overline{\mathbf{\Delta}}\right)  ^{\prime}\mathbf{P}_{b}%
(\widetilde{\mathbf{W}}_{\epsilon}^{\ast}(r))^{-1}\left(  \sqrt{\mathbf{12}%
}\int_{0}^{1}(s-\frac{1}{2})d\mathbf{W}_{\epsilon}(s)+\frac{1}{\sqrt
{\mathbf{12}}}\mathbf{\Lambda}_{\epsilon}^{\ast-1}\overline{\mathbf{\Delta}%
}\right)  ,
\]%
\[
=\left(  \mathbf{Z}_{\epsilon}^{\ast}+\frac{1}{\sqrt{\mathbf{12}}%
}\mathbf{\Lambda}_{\epsilon}^{\ast-1}\overline{\mathbf{\Delta}}\right)
^{\prime}\mathbf{P}_{b}(\widetilde{\mathbf{W}}_{\epsilon}^{\ast}%
(r))^{-1}\left(  \mathbf{Z}_{\epsilon}^{\ast}+\frac{1}{\sqrt{\mathbf{12}}%
}\mathbf{\Lambda}_{\epsilon}^{\ast-1}\overline{\mathbf{\Delta}}\right)  .
\]
It is not difficult to show that $\mathbf{Z}_{\epsilon}^{\ast}=\sqrt
{\mathbf{12}}\int_{0}^{1}(s-\frac{1}{2})d\mathbf{W}_{\epsilon}(s)$ is
distributed $N(\mathbf{0,I}_{q})$ and is independent of $\widetilde{\mathbf{W}%
}_{\epsilon}^{\ast}(r)$.

For the case of very small slopes, the arguments are similar as the large to
small slopes case and details are omitted.

For the case of zero trend slopes, different arguments are needed because
$\mathbf{\theta}$ is not defined and $\overline{\mathbf{\Delta}}=\mathbf{0}$.
Write the Wald statistic as%
\[
Wald_{IV}=\left(  \mathbf{R}\widehat{\mathbf{\theta}}-\mathbf{r}\right)
^{\prime}\left[  \mathbf{R}\left(  T^{-3}%
{\displaystyle\sum\limits_{t=1}^{T}}
(t-\overline{t})^{2}\right)  \left(  T^{-3/2}\widehat{\mathbf{D}}_{2}\right)
^{-1}\widehat{\mathbf{\Omega}}_{\mathbf{\epsilon}}\left(  T^{-3/2}%
\widehat{\mathbf{D}}_{2}\right)  ^{-1}\mathbf{R}^{\prime}\right]  ^{-1}\left(
\mathbf{R}\widehat{\mathbf{\theta}}-\mathbf{r}\right)  .
\]
Using Theorem 1 it follows that%
\begin{equation}
\mathbf{R}\widehat{\mathbf{\theta}}-\mathbf{r\Rightarrow RD}_{\mathbf{B}_{u2}%
}^{-1}\int_{0}^{1}\left(  s-\frac{1}{2}\right)  d\mathbf{B}_{\mathbf{u}%
1}(s)-\mathbf{r,} \label{rthetahat}%
\end{equation}
and using Lemma 4 it follows that%
\begin{equation}
T^{-3/2}\widehat{\mathbf{D}}_{2}\mathbf{\Rightarrow D}_{\mathbf{B}_{u2}}.
\label{dhat2}%
\end{equation}
Because the slopes are zero and $\mathbf{\theta}$ is not defined, it now
follows that%
\[
T^{-1/2}\sum_{t=1}^{[rT]}\widehat{\scalebox{1.5}{$\bm{\epsilon}$}}_{\theta
t}=T^{-1/2}%
{\displaystyle\sum\limits_{t=1}^{[rT]}}
\left(  \mathbf{y}_{1t}-\overline{\mathbf{y}}_{1}\right)  -\left(  T^{-1/2}%
{\displaystyle\sum\limits_{t=1}^{[rT]}}
\mathbf{H}_{t}\right)  \widehat{\mathbf{\theta}}=T^{-1/2}%
{\displaystyle\sum\limits_{t=1}^{[rT]}}
\left(  \mathbf{u}_{1t}-\overline{\mathbf{u}}_{1}\right)  -\left(  T^{-1/2}%
{\displaystyle\sum\limits_{t=1}^{[rT]}}
\mathbf{H}_{t}\right)  \widehat{\mathbf{\theta}},
\]
where the $i^{th}$ diagonal element of $\mathbf{H}_{t}$ is given by
$u_{2t}^{(i)}-\overline{u}_{2}^{(i)}$. Using Lemma 1 and Theorem 1, the
following holds%
\begin{equation}
T^{-1/2}\sum_{t=1}^{[rT]}\widehat{\scalebox{1.5}{$\bm{\epsilon}$}}_{\theta
t}\mathbf{\Rightarrow B}_{\mathbf{u}1}(s)-r\mathbf{B}_{\mathbf{u}%
1}(1)-\mathbf{D}_{\widehat{\mathbf{B}}_{u2}(r)}\mathbf{D}_{\mathbf{B}_{u2}%
}^{-1}\int_{0}^{1}\left(  s-\frac{1}{2}\right)  d\mathbf{B}_{\mathbf{u}%
1}(s)\equiv\mathbf{H}_{2}(r). \label{epshat}%
\end{equation}
Using (\ref{rthetahat}), (\ref{dhat2}), and (\ref{epshat}), the limit of
$Wald_{IV}$ is easily obtained for the case where slopes are zero.

\bigskip

\bigskip

\noindent\textbf{Proof of Theorem 3.} Begin with some algebraic calculations
for $g_{\widehat{\beta}}$:%
\begin{align*}
g_{\widehat{\beta}}  &  =\widehat{\beta}_{2}^{(2)}\widehat{\beta}_{1}%
^{(1)}-\widehat{\beta}_{2}^{(1)}\widehat{\beta}_{1}^{(2)}\\
&  =\left[  \beta_{2}^{(2)}+\left(  \widehat{\beta}_{2}^{(2)}-\beta_{2}%
^{(2)}\right)  \right]  \left[  \beta_{1}^{(1)}+\left(  \widehat{\beta}%
_{1}^{(1)}-\beta_{1}^{(1)}\right)  \right]  -\left[  \beta_{2}^{(1)}+\left(
\widehat{\beta}_{2}^{(1)}-\beta_{2}^{(1)}\right)  \right]  \left[  \beta
_{1}^{(2)}+\left(  \widehat{\beta}_{1}^{(2)}-\beta_{1}^{(2)}\right)  \right]
\\
&  =\beta_{2}^{(2)}\beta_{1}^{(1)}-\beta_{2}^{(1)}\beta_{1}^{(2)}%
+\mathbf{R}_{\beta}\left(  \widehat{\mathbf{\beta}}-\mathbf{\beta}\right)
+\left(  \widehat{\beta}_{2}^{(2)}-\beta_{2}^{(2)}\right)  \left(
\widehat{\beta}_{1}^{(1)}-\beta_{1}^{(1)}\right)  -\left(  \widehat{\beta}%
_{2}^{(1)}-\beta_{2}^{(1)}\right)  \left(  \widehat{\beta}_{1}^{(2)}-\beta
_{1}^{(2)}\right)  .
\end{align*}
Replacing $\beta_{2}^{(2)}\beta_{1}^{(1)}-\beta_{2}^{(1)}\beta_{1}^{(2)}$ with
$\beta_{2}^{(1)}\beta_{2}^{(2)}\overline{\mathbf{\Delta}}T^{-3/2+\kappa}$
using (\ref{h1prod}) gives%
\[
g_{\widehat{\beta}}=\beta_{2}^{(1)}\beta_{2}^{(2)}\overline{\mathbf{\Delta}%
}T^{-3/2+\kappa}+\mathbf{R}_{\beta}\left(  \widehat{\mathbf{\beta}%
}-\mathbf{\beta}\right)  +\left(  \widehat{\beta}_{2}^{(2)}-\beta_{2}%
^{(2)}\right)  \left(  \widehat{\beta}_{1}^{(1)}-\beta_{1}^{(1)}\right)
-\left(  \widehat{\beta}_{2}^{(1)}-\beta_{2}^{(1)}\right)  \left(
\widehat{\beta}_{1}^{(2)}-\beta_{1}^{(2)}\right)  .
\]
Replacing $\beta_{2}^{(1)}$ and $\beta_{2}^{(2)}$ with their local values
$T^{-\kappa}\overline{\beta}_{2}^{(1)}$ and $T^{-\kappa}\overline{\beta}%
_{2}^{(2)}$ gives%
\[
g_{\widehat{\beta}}=\overline{\beta}_{2}^{(1)}\overline{\beta}_{2}%
^{(2)}\overline{\mathbf{\Delta}}T^{-3/2-\kappa}+T^{-\kappa}\mathbf{R}%
_{\overline{\beta}}\left(  \widehat{\mathbf{\beta}}-\mathbf{\beta}\right)
+\left(  \widehat{\beta}_{2}^{(2)}-\beta_{2}^{(2)}\right)  \left(
\widehat{\beta}_{1}^{(1)}-\beta_{1}^{(1)}\right)  -\left(  \widehat{\beta}%
_{2}^{(1)}-\beta_{2}^{(1)}\right)  \left(  \widehat{\beta}_{1}^{(2)}-\beta
_{1}^{(2)}\right)  .
\]
Recall that%
\[
t_{prod}=\frac{g_{\widehat{\beta}}}{\sqrt{\mathbf{R}_{\widehat{\beta}%
}\widehat{\mathbf{\Omega}}_{\mathbf{u}}\mathbf{R}_{\widehat{\beta}}^{\prime
}\left(
{\textstyle\sum\nolimits_{t=1}^{T}}
(t-\overline{t})^{2}\right)  ^{-1}}}.
\]
Because $\widehat{\mathbf{\Omega}}_{\mathbf{u}}$ is exactly invariant to the
values of the trend slopes, Lemma 1 and standard fixed-$b$ arguments give%
\[
\widehat{\mathbf{\Omega}}_{\mathbf{u}}\mathbf{\Rightarrow
P(\widetilde{\mathbf{B}}_{u}\mathbf{(r)})=\Lambda}_{\mathbf{u}}%
\mathbf{P(\widetilde{\mathbf{W}}_{u}\mathbf{(r)})\Lambda}_{\mathbf{u}}%
^{\prime}.
\]
For the case of large to small trend slopes, it follows from Lemma 1 that%
\begin{align*}
T^{3/2+\kappa}g_{\widehat{\beta}}  &  =\overline{\beta}_{2}^{(1)}%
\overline{\beta}_{2}^{(2)}\overline{\mathbf{\Delta}}+\mathbf{R}_{\overline
{\beta}}T^{3/2}\left(  \widehat{\mathbf{\beta}}-\mathbf{\beta}\right)
+T^{3/2+\kappa}\left(  \widehat{\beta}_{2}^{(2)}-\beta_{2}^{(2)}\right)
\left(  \widehat{\beta}_{1}^{(1)}-\beta_{1}^{(1)}\right) \\
&  -T^{3/2+\kappa}\left(  \widehat{\beta}_{2}^{(1)}-\beta_{2}^{(1)}\right)
\left(  \widehat{\beta}_{1}^{(2)}-\beta_{1}^{(2)}\right)  ,\\
&  =\overline{\beta}_{2}^{(1)}\overline{\beta}_{2}^{(2)}\overline
{\mathbf{\Delta}}+\mathbf{R}_{\overline{\beta}}T^{3/2}\left(
\widehat{\mathbf{\beta}}-\mathbf{\beta}\right)  +o_{p}(1),\\
&  \mathbf{\Rightarrow}\text{ }\overline{\beta}_{2}^{(1)}\overline{\beta}%
_{2}^{(2)}\overline{\mathbf{\Delta}}+12\mathbf{R}_{\overline{\beta}%
}\mathbf{\Lambda}_{u}\int_{0}^{1}(s-\frac{1}{2})d\mathbf{W}_{u}(s).
\end{align*}
Note that the $T^{3/2+\kappa}\left(  \widehat{\beta}_{2}^{(i)}-\beta_{2}%
^{(i)}\right)  \left(  \widehat{\beta}_{1}^{(j)}-\beta_{1}^{(j)}\right)
=T^{-3/2+\kappa}T^{3/2}\left(  \widehat{\beta}_{2}^{(i)}-\beta_{2}%
^{(i)}\right)  T^{3/2}\left(  \widehat{\beta}_{1}^{(j)}-\beta_{1}%
^{(j)}\right)  $ terms are $o_{p}(1)$ because $T^{-3/2+\kappa}\rightarrow0$ as
$T\rightarrow\infty$ for $0\leq\kappa<\frac{3}{2}$.

Collecting limits gives%
\[
t_{prod}=\frac{g_{\widehat{\beta}}}{\sqrt{\mathbf{R}_{\widehat{\beta}%
}\widehat{\mathbf{\Omega}}_{\mathbf{u}}\mathbf{R}_{\widehat{\beta}}^{\prime
}\left(
{\textstyle\sum\nolimits_{t=1}^{T}}
(t-\overline{t})^{2}\right)  ^{-1}}}=\frac{T^{3/2+\kappa}g_{\widehat{\beta}}%
}{\sqrt{T^{\kappa}\mathbf{R}_{\widehat{\beta}}\widehat{\mathbf{\Omega}%
}_{\mathbf{u}}\mathbf{R}_{\widehat{\beta}}^{\prime}T^{\kappa}\left(  T^{-3}%
{\textstyle\sum\nolimits_{t=1}^{T}}
(t-\overline{t})^{2}\right)  ^{-1}}}%
\]%
\begin{align*}
&  \mathbf{\Rightarrow}\text{ }\frac{\overline{\beta}_{2}^{(1)}\overline
{\beta}_{2}^{(2)}\overline{\mathbf{\Delta}}+12\mathbf{R}_{\overline{\beta}%
}\mathbf{\Lambda}_{u}\int_{0}^{1}(s-\frac{1}{2})d\mathbf{W}_{u}(s)}%
{\sqrt{12\mathbf{R}_{\overline{\beta}}\mathbf{\Lambda}_{\mathbf{u}%
}\mathbf{P(\widetilde{\mathbf{W}}_{u}\mathbf{(r)})\Lambda}_{\mathbf{u}%
}^{\prime}\mathbf{R}_{\overline{\beta}}^{\prime}}}\equiv\frac{\overline{\beta
}_{2}^{(1)}\overline{\beta}_{2}^{(2)}\overline{\mathbf{\Delta}}%
+12\mathbf{\Lambda}_{u}^{\ast}\int_{0}^{1}(s-\frac{1}{2})dw_{u}^{\ast}%
(s))}{\sqrt{12\mathbf{\Lambda}_{u}^{\ast2}P_{b}(\widetilde{w}_{u}^{\ast}(r))}%
},\\
&  =\frac{\frac{\overline{\beta}_{2}^{(1)}\overline{\beta}_{2}^{(2)}%
\overline{\mathbf{\Delta}}}{\mathbf{\Lambda}_{u}^{\ast}\sqrt{12}}+\sqrt
{12}\int_{0}^{1}(s-\frac{1}{2})dw_{u}^{\ast}(s))}{\sqrt{P_{b}(\widetilde{w}%
_{u}^{\ast}(r))}}\equiv\frac{\frac{\overline{\beta}_{2}^{(1)}\overline{\beta
}_{2}^{(2)}\overline{\mathbf{\Delta}}}{\mathbf{\Lambda}_{u}^{\ast}\sqrt{12}%
}+Z_{u}^{\ast}}{\sqrt{P_{b}(\widetilde{w}_{u}^{\ast}(r))}}.
\end{align*}
using Lemmas 1 and 2. Replacing $\mathbf{R}_{\overline{\beta}}\mathbf{\Lambda
}_{u}\mathbf{W}_{u}(s)$ with $\mathbf{\Lambda}_{u}^{\ast}w_{u}^{\ast}(s)$
where $w_{u}^{\ast}(s)$ is a univariate Wiener process and $\mathbf{\Lambda
}_{u}^{\ast}=\sqrt{\mathbf{R}_{\overline{\beta}}\mathbf{\Lambda}_{\mathbf{u}%
}\mathbf{\Lambda}_{\mathbf{u}}^{\prime}\mathbf{R}_{\overline{\beta}}^{\prime}%
}$, gives the second to last expression of the limit. It is easy to show that
$Z_{u}^{\ast}=\sqrt{12}\int_{0}^{1}(s-\frac{1}{2})dw_{u}^{\ast}(s))$ is
distributed $N(0,1)$ and is independent of $\widetilde{w}_{u}^{\ast}(r)$. For
the case of very small trend slopes $(\kappa=\frac{3}{2})$, it follows that%
\begin{align*}
T^{3}g_{\widehat{\beta}}  &  =\overline{\beta}_{2}^{(1)}\overline{\beta}%
_{2}^{(2)}\overline{\mathbf{\Delta}}+\mathbf{R}_{\overline{\beta}}%
T^{3/2}\left(  \widehat{\mathbf{\beta}}-\mathbf{\beta}\right)  +T^{3/2}\left(
\widehat{\beta}_{2}^{(2)}-\beta_{2}^{(2)}\right)  T^{3/2}\left(
\widehat{\beta}_{1}^{(1)}-\beta_{1}^{(1)}\right) \\
&  -T^{3/2}\left(  \widehat{\beta}_{2}^{(1)}-\beta_{2}^{(1)}\right)
T^{3/2}\left(  \widehat{\beta}_{1}^{(2)}-\beta_{1}^{(2)}\right) \\
&  \mathbf{\Rightarrow}\overline{\beta}_{2}^{(1)}\overline{\beta}_{2}%
^{(2)}\overline{\mathbf{\Delta}}+12\mathbf{R}_{\overline{\beta}}%
\mathbf{\Lambda}_{u}\int_{0}^{1}(s-\frac{1}{2})d\mathbf{W}_{u}(s)+\Psi
_{2}^{(2)}\Psi_{1}^{(1)}-\Psi_{2}^{(1)}\Psi_{1}^{(2)}.
\end{align*}
Using Lemmas 1 and 3, the limit of $t_{prod}$ is given by%
\begin{align*}
t_{prod}  &  =\frac{g_{\widehat{\beta}}}{\sqrt{\mathbf{R}_{\widehat{\beta}%
}\widehat{\mathbf{\Omega}}_{\mathbf{u}}\mathbf{R}_{\widehat{\beta}}^{\prime
}\left(
{\textstyle\sum\nolimits_{t=1}^{T}}
(t-\overline{t})^{2}\right)  ^{-1}}}=\frac{T^{3}g_{\widehat{\beta}}}%
{\sqrt{T^{3/2}\mathbf{R}_{\widehat{\beta}}\widehat{\mathbf{\Omega}%
}_{\mathbf{u}}\mathbf{R}_{\widehat{\beta}}^{\prime}T^{3/2}\left(  T^{-3}%
{\textstyle\sum\nolimits_{t=1}^{T}}
(t-\overline{t})^{2}\right)  ^{-1}}}\\
&  \mathbf{\Rightarrow}\frac{\overline{\beta}_{2}^{(1)}\overline{\beta}%
_{2}^{(2)}\overline{\mathbf{\Delta}}+12\mathbf{R}_{\overline{\beta}%
}\mathbf{\Lambda}_{u}\int_{0}^{1}(s-\frac{1}{2})d\mathbf{W}_{u}(s)+\Psi
_{2}^{(2)}\Psi_{1}^{(1)}-\Psi_{2}^{(1)}\Psi_{1}^{(2)}}{\sqrt{12\left(
\mathbf{R}_{\overline{\beta}}+\left[  \Psi_{2}^{(2)},-\Psi_{2}^{(1)},-\Psi
_{1}^{(2)},\Psi_{1}^{(1)}\right]  \right)  \mathbf{P(\widetilde{\mathbf{B}%
}_{u}\mathbf{(r)})}\left(  \mathbf{R}_{\overline{\beta}}+\left[  \Psi
_{2}^{(2)},-\Psi_{2}^{(1)},-\Psi_{1}^{(2)},\Psi_{1}^{(1)}\right]  \right)
^{\prime}}},
\end{align*}
as required. The limit of $t_{prod}$ for the case of zero trend slopes is
obtained by replacing $\overline{\beta}_{2}^{(1)}$, $\overline{\beta}%
_{2}^{(2)}$ and $\mathbf{R}_{\overline{\beta}}$ with zeros.

\section{Acknowledgements}

I thank Ross McKitrick for detailed comments on earlier drafts of the paper. I
thank seminar and conference participants at U. Colorado Boulder, Indiana
University (conference in honor of Joon Park), Guelph University, and Michigan
State University for useful comments and suggestions. I thank John Christy for
providing the data used in the empirical application.

\bibliographystyle{kluwer}
\bibliography{trendratio}

\bigskip

\bigskip

\scalebox{1.0}{
\begin{tabular}
[c]{ccccccccccccc}\multicolumn{13}{c}{Table 1: Empirical Null Rejection Probabilities,
$H_{0}:\theta^{(1)}=\theta^{(2)}$, 5\% Nominal Level}\\
\multicolumn{13}{c}{For $i=1,2,$ $\beta_{1}^{(i)}=\theta^{(i)}\beta_{2}^{(i)}$. If $\beta_{2}^{(i)}=0$, then $\beta_{1}^{(i)}=0$. 10,000 Replications}\\\hline\hline
\multicolumn{5}{c}{iid Noise} & $b=A91$ &  & $b=0.25$ &  & $b=0.5$ &  &
$b=1.0$ & \\\hline
$T$ & $\beta_{2}^{(1)}$ & $\theta^{(1)}$ & $\beta_{2}^{(2)}$ & $\theta^{(2)}$
& ${\small t}_{IV}$ & $t_{prod}$ & ${\small t}_{IV}$ & $t_{prod}$ &
${\small t}_{IV}$ & $t_{prod}$ & ${\small t}_{IV}$ & $t_{prod}$\\\hline
50 & 10 & 1 & 10 & 1 & .047 & .051 & .050 & .050 & .050 & .050 & .052 & .052\\
& 2 & 1 & 2 & 1 & .047 & .051 & .050 & .050 & .051 & .050 & .052 & .052\\
& .2 & 1 & .2 & 1 & .045 & .051 & .049 & .049 & .049 & .050 & .052 & .053\\
& .05 & 1 & .05 & 1 & .017 & .045 & .032 & .047 & .039 & .050 & .045 & .049\\
& .025 & 1 & .025 & 1 & .004 & .033 & .017 & .041 & .026 & .045 & .032 &
.049\\
& .005 & 1 & .005 & 1 & .000 & .002 & .005 & .011 & .010 & .019 & .012 &
.023\\
& 0 & na & 0 & na & .000 & .001 & .004 & .009 & .008 & .017 & .012 &
.021\\\hline
100 & 10 & 1 & 10 & 1 & .052 & .055 & .052 & .052 & .050 & .050 & .053 &
.053\\
& 2 & 1 & 2 & 1 & .052 & .055 & .052 & .052 & .050 & .050 & .053 & .053\\
& .2 & 1 & .2 & 1 & .052 & .054 & .053 & .052 & .050 & .050 & .053 & .054\\
& .05 & 1 & .05 & 1 & .048 & .053 & .050 & .051 & .052 & .051 & .052 & .055\\
& .025 & 1 & .025 & 1 & .030 & .050 & .045 & .052 & .049 & .052 & .049 &
.054\\
& .005 & 1 & .005 & 1 & .002 & .012 & .010 & .031 & .015 & .035 & .020 &
.039\\
& 0 & na & 0 & na & .000 & .001 & .003 & .010 & .008 & .017 & .010 &
.021\\\hline
200 & 10 & 1 & 10 & 1 & .048 & .049 & .050 & .050 & .049 & .049 & .053 &
.053\\
& 2 & 1 & 2 & 1 & .048 & .049 & .050 & .050 & .049 & .049 & .053 & .053\\
& .2 & 1 & .2 & 1 & .048 & .049 & .050 & .049 & .049 & .050 & .053 & .052\\
& .05 & 1 & .05 & 1 & .047 & .049 & .049 & .049 & .050 & .051 & .053 & .053\\
& .025 & 1 & .025 & 1 & .045 & .048 & .048 & .050 & .050 & .051 & .053 &
.054\\
& .005 & 1 & .005 & 1 & .008 & .041 & .025 & .047 & .035 & .047 & .039 &
.051\\
& 0 & na & 0 & na & .000 & .000 & .003 & .009 & .007 & .016 & .009 &
.018\\\hline
&  &  &  &  &  &  &  &  &  &  &  & \\
\multicolumn{5}{c}{Serially Correlated Noise} & \multicolumn{2}{c}{$b=A91$} &
\multicolumn{2}{c}{$b=0.25$} & \multicolumn{2}{c}{$b=0.5$} &
\multicolumn{2}{c}{$b=1.0$}\\\hline
$T$ & $\beta_{2}^{(1)}$ & $\theta^{(1)}$ & $\beta_{2}^{(2)}$ & $\theta^{(2)}$
& ${\small t}_{IV}$ & $t_{prod}$ & ${\small t}_{IV}$ & $t_{prod}$ &
${\small t}_{IV}$ & $t_{prod}$ & ${\small t}_{IV}$ & $t_{prod}$\\\hline
50 & 10 & 1 & 10 & 1 & .123 & .160 & .087 & .086 & .062 & .063 & .066 & .066\\
& 2 & 1 & 2 & 1 & .122 & .160 & .087 & .087 & .064 & .063 & .066 & .066\\
& .2 & 1 & .2 & 1 & .093 & .162 & .065 & .087 & .055 & .063 & .058 & .066\\
& .05 & 1 & .05 & 1 & .049 & .152 & .033 & .066 & .031 & .055 & .035 & .055\\
& .025 & 1 & .025 & 1 & .026 & .078 & .020 & .038 & .022 & .041 & .026 &
.043\\
& .005 & 1 & .005 & 1 & .005 & .019 & .007 & .016 & .011 & .021 & .014 &
.023\\
& 0 & na & 0 & na & .005 & .016 & .006 & .014 & .010 & .019 & .015 &
.021\\\hline
100 & 10 & 1 & 10 & 1 & .100 & .132 & .062 & .062 & .056 & .055 & .060 &
.060\\
& 2 & 1 & 2 & 1 & .101 & .132 & .062 & .062 & .056 & .055 & .060 & .060\\
& .2 & 1 & .2 & 1 & .096 & .131 & .058 & .062 & .055 & .056 & .058 & .060\\
& .05 & 1 & .05 & 1 & .073 & .142 & .040 & .061 & .038 & .057 & .040 & .061\\
& .025 & 1 & .025 & 1 & .056 & .147 & .034 & .056 & .035 & .055 & .039 &
.060\\
& .005 & 1 & .005 & 1 & .010 & .028 & .011 & .018 & .015 & .024 & .019 &
.028\\
& 0 & na & 0 & na & .003 & .009 & .007 & .012 & .010 & .019 & .011 &
.020\\\hline
200 & 10 & 1 & 10 & 1 & .085 & .114 & .054 & .054 & .052 & .052 & .056 &
.056\\
& 2 & 1 & 2 & 1 & .085 & .114 & .054 & .054 & .053 & .052 & .056 & .056\\
& .2 & 1 & .2 & 1 & .085 & .114 & .054 & .053 & .053 & .053 & .057 & .056\\
& .05 & 1 & .05 & 1 & .079 & .117 & .050 & .053 & .049 & .054 & .042 & .056\\
& .025 & 1 & .025 & 1 & .081 & .122 & .046 & .052 & .044 & .052 & .046 &
.053\\
& .005 & 1 & .005 & 1 & .045 & .107 & .027 & .033 & .031 & .039 & .034 &
.041\\
& 0 & na & 0 & na & .002 & .005 & .006 & .010 & .009 & .017 & .010 &
.019\\\hline\hline
\end{tabular}
}\newpage

\scalebox{1.0}{
\begin{tabular}
[c]{cccccccccccccc}\multicolumn{14}{c}{Table 2: Finite Sample Power, 5\% Nominal Level, $T=100$,
Two-sided Tests.}\\
\multicolumn{14}{c}{10,000 Replications, $H_{0}:\theta^{(1)}=\theta^{(2)},$
$H_{1}:\theta^{(1)}\neq\theta^{(2)}$, $\beta_{1}^{(i)}=\theta^{(i)}\beta
_{2}^{(i)}$ for $i=1,2.$}\\\hline\hline
\multicolumn{5}{c}{Serially Correlated Noise} & \multicolumn{2}{c}{$b=A91$} &
\multicolumn{2}{c}{$b=0.25$} & \multicolumn{2}{c}{$b=0.5$} &
\multicolumn{2}{c}{$b=1.0$} & \\\cline{1-5}\cline{1-7}\cline{8-14}
& $\beta_{2}^{(1)}$ & $\theta^{(1)}$ & $\beta_{2}^{(2)}$ & $\theta^{(2)}$ &
${\small t}_{IV}$ & $t_{prod}$ & ${\small t}_{IV}$ & $t_{prod}$ &
${\small t}_{IV}$ & $t_{prod}$ & ${\small t}_{IV}$ & $t_{prod}$ & \\\hline
& 10 & 1.0 & 10 & .993 & .727 & .800 & .586 & .583 & .414 & .412 & .353 &
.351 & \\
&  &  &  & .995 & .482 & .558 & .344 & .341 & .243 & .242 & .224 & .224 & \\
&  &  &  & .998 & .221 & .275 & .143 & .143 & .111 & .110 & .110 & .110 & \\
&  &  &  & \textbf{1} & \textbf{.100} & \textbf{.132} & \textbf{.062} &
\textbf{.062} & \textbf{.056} & \textbf{.055} & \textbf{.060} & \textbf{.060}
& \\
&  &  &  & 1.002 & .205 & .258 & .134 & .136 & .102 & .102 & .103 & .103 & \\
&  &  &  & 1.005 & .466 & .547 & .330 & .333 & .234 & .236 & .214 & .215 & \\
&  &  &  & 1.008 & .712 & .797 & .568 & .571 & .395 & .397 & .338 & .341 &
\\\hline
& 2 & 1.000 & 2 & .960 & .780 & .834 & .650 & .634 & .467 & .448 & .391 &
.378 & \\
&  &  &  & .973 & .534 & .597 & .389 & .376 & .277 & .368 & .251 & .243 & \\
&  &  &  & .987 & .243 & .291 & .157 & .151 & .121 & .117 & .119 & .114 & \\
&  &  &  & \textbf{1} & \textbf{.101} & \textbf{.132} & \textbf{.062} &
\textbf{.062} & \textbf{.056} & \textbf{.055} & \textbf{.060} & \textbf{.060}
& \\
&  &  &  & 1.013 & .212 & .274 & .139 & .144 & .106 & .109 & .106 & .108 & \\
&  &  &  & 1.027 & .485 & .581 & .348 & .363 & .247 & .257 & .224 & .231 & \\
&  &  &  & 1.040 & .731 & .829 & .595 & .612 & .413 & .429 & .353 & .366 &
\\\hline
& .2 & 1.0 & .2 & .600 & .936 & .858 & .832 & .662 & .651 & .471 & .547 &
.398 & \\
&  &  &  & .733 & .698 & .625 & .523 & .392 & .382 & .273 & .327 & .247 & \\
&  &  &  & .867 & .322 & .306 & .205 & .153 & .156 & .118 & .145 & .115 & \\
&  &  &  & \textbf{1} & \textbf{.096} & \textbf{.131} & \textbf{.058} &
\textbf{.062} & \textbf{.055} & \textbf{.056} & \textbf{.058} & \textbf{.060}
& \\
&  &  &  & 1.133 & .132 & .260 & .087 & .139 & .075 & .105 & .080 & .108 & \\
&  &  &  & 1.267 & .314 & .550 & .217 & .340 & .159 & .241 & .158 & .221 & \\
&  &  &  & 1.400 & .499 & .792 & .378 & .567 & .264 & .394 & .243 & .342 &
\\\hline
& .05 & 1.0 & .05 & -1.00 & .725 & .925 & .582 & .728 & .425 & .521 & .366 &
.437 & \\
&  &  &  & -.333 & .821 & .785 & .661 & .514 & .484 & .359 & .410 & .317 & \\
&  &  &  & .333 & .658 & .430 & .422 & .198 & .300 & .150 & .260 & .146 & \\
&  &  &  & \textbf{1} & \textbf{.073} & \textbf{.142} & \textbf{.040} &
\textbf{.061} & \textbf{.038} & \textbf{.057} & \textbf{.040} & \textbf{.061}
& \\
&  &  &  & 1.667 & .022 & .275 & .018 & .152 & .026 & .118 & .030 & .114 & \\
&  &  &  & 2.333 & .074 & .527 & .054 & .324 & .058 & .234 & .063 & .214 & \\
&  &  &  & 3.000 & .139 & .696 & .095 & .463 & .086 & .331 & .092 & .295 &
\\\hline
& .025 & 1.0 & .025 & -14 & .229 & .637 & .165 & .412 & .127 & .297 & .121 &
.267 & \\
&  &  &  & -9 & .243 & .650 & .174 & .420 & .133 & .302 & .125 & .269 & \\
&  &  &  & -4 & .282 & .671 & .193 & .415 & .147 & .295 & .141 & .258 & \\
&  &  &  & \textbf{1} & \textbf{.056} & \textbf{.147} & \textbf{.034} &
\textbf{.056} & \textbf{.035} & \textbf{.055} & \textbf{.039} & \textbf{.060}
& \\
&  &  &  & 6 & .075 & .447 & .053 & .275 & .051 & .206 & .056 & .192 & \\
&  &  &  & 11 & .120 & .523 & .083 & .336 & .070 & .244 & .074 & .224 & \\
&  &  &  & 16 & .143 & .547 & .098 & .355 & .080 & .258 & .083 & .232 &
\\\hline
& .005 & 2 & .005 & -49 & .077 & .098 & .055 & .067 & .046 & .062 & .051 &
.063 & \\
&  &  &  & -32.33 & .073 & .099 & .049 & .066 & .044 & .061 & .047 & .064 & \\
&  &  &  & -15.67 & .063 & .100 & .040 & .065 & .038 & .060 & .043 & .062 & \\
&  &  &  & \textbf{1} & \textbf{.010} & \textbf{.028} & \textbf{.011} &
\textbf{.018} & \textbf{.015} & \textbf{.024} & \textbf{.019} & \textbf{.028}
& \\
&  &  &  & 17.67 & .047 & .094 & .036 & .061 & .036 & .059 & .037 & .061 & \\
&  &  &  & 34.33 & .062 & .095 & .047 & .065 & .041 & .061 & .046 & .062 & \\
&  &  &  & 51 & .070 & .095 & .051 & .065 & .045 & .060 & .049 & .063 &
\\\hline\hline
\end{tabular}
}\newpage

\scalebox{1.0}{
\begin{tabular}
[c]{ccccccc}\multicolumn{7}{c}{Table 3: Estimated Linear Trend Slopes in Degrees Celsius
per Decade}\\
\multicolumn{7}{c}{Observed Temperatures by hPa Level, Annual Data 1958-2024,
$T=67$}\\\hline\hline
hPa &  & RICH & RAOB & RATP & ERA5 & JRA3Q\\\hline
&  &  &  &  &  & \\
SFC & $\widehat{\beta}$ & .142 & .143 & .143 & .132 & .145\\
& CI & (.117, .167) & (.118, .167) & (.118, .167) & (.103, .160) & (.121,
.168)\\\hline
&  &  &  &  &  & \\
850 & $\widehat{\beta}$ & .117 & .081 & .121 & .115 & .137\\
& CI & (.086, .148) & (.050, .112) & (.086, .156) & (.083, .148) & (.111,
.164)\\\hline
&  &  &  &  &  & \\
700 & $\widehat{\beta}$ & .195 & .152 & .134 & .169 & .166\\
& CI & (.158, .232) & (.117, .187) & (.100, .168) & (.133, .205) & (.132,
.199)\\\hline
&  &  &  &  &  & \\
500 & $\widehat{\beta}$ & .204 & .147 & .158 & .147 & .168\\
& CI & (.167, .240) & (.108, .185) & (.121, .186) & (.111, .182) & (.133,
.202)\\\hline
&  &  &  &  &  & \\
400 & $\widehat{\beta}$ & .220 & .179 & .200 & .192 & .189\\
& CI & (.177, .264) & (.135, .224) & (.156, .244) & (.150, .235) & (.145,
.234)\\\hline
&  &  &  &  &  & \\
300 & $\widehat{\beta}$ & .230 & .193 & .231 & .221 & .259\\
& CI & (.178, .282) & (.139, .247) & (.183, .280) & (.172, .270) & (.209,
.309)\\\hline
&  &  &  &  &  & \\
200 & $\widehat{\beta}$ & .225 & .197 & .159 & .226 & .266\\
& CI & (.169, .280) & (.145, .249) & (.109, .209) & (.176, .276) & (.219,
.313)\\\hline
&  &  &  &  &  & \\
150 & $\widehat{\beta}$ & .196 & .107 & .060 & .122 & .151\\
& CI & (.139, .252) & (.048, .166) & (.017, .104) & (.066, .178) & (.107,
.195)\\\hline
&  &  &  &  &  & \\
100 & $\widehat{\beta}$ & -.041 & -.190 & -.157 & -.176 & -.247\\
& CI & (-.126, .045) & (-.358, -.022) & (-.221, -.093) & (-.343, -.008) &
(-.358, -.135)\\\hline
&  &  &  &  &  & \\
70 & $\widehat{\beta}$ & -.310 & -.374 & -.442 & -.321 & -.396\\
& CI & (-.471, -.149) & (-515, -.233) & (-585, -.298) & (-.462, -.180) &
(-.539, -.253)\\\hline
&  &  &  &  &  & \\
50 & $\widehat{\beta}$ & -.468 & -.391 & -.439 & -.339 & -.329\\
& CI & (-.609, -.328) & (-.515, -.267) & (-.538, -.341) & (-.440, -.239) &
(-.422, -.236)\\\hline
&  &  &  &  &  & \\
30 & $\widehat{\beta}$ & -.409 & -.382 & -.514 & -.353 & -.422\\
& CI & (-.488, -.329) & (-.461, -.303) & (-.594, -.434) & (-.435, -.271) &
(-.487, -.358)\\\hline
&  &  &  &  &  & \\
20 & $\widehat{\beta}$ & -.374 & -.397 & NA & -.361 & -.568\\
& CI & (-.450, -.299) & (-.469, -.325) &  & (-.428, -.293) & (-637,
-.500)\\\hline\hline
\multicolumn{7}{l}{Note: 95\% fixed-$b$ confidence intervals in brackets using
Daniell $k(x)$ function with Andrews}\\
\multicolumn{7}{l}{(1991) data dependent bandwidth. Data dependent bandwidth
sample size ratios}\\
\multicolumn{7}{l}{ranged from .024 to .136.}\end{tabular}
}\newpage\scalebox{1.0}{
\begin{tabular}
[c]{ccccccc}\multicolumn{7}{c}{Table 4: IV Estimated Trend Slope Ratios Relative to
Surface by hPa Levels}\\
\multicolumn{7}{c}{Annual Data 1958-2024, $T=67$}\\\hline\hline
hPa/SFC &  & RICH & RAOB & RATP & ERA5 & JRA3Q\\\hline
&  &  &  &  &  & \\
850/SFC & $\widehat{\theta}$ & .823 & .568 & .849 & .875 & .950\\
& CI & (.677, .954) & (.409, .692) & (.703, .971) & (.750, .985) & (.896,
.997)\\\hline
&  &  &  &  &  & \\
700/SFC & $\widehat{\theta}$ & 1.37 & 1.07 & .941 & 1.28 & 1.15\\
& CI & (1.26, 1.49) & (.947, 1.17) & (.821, 1.05) & (1.16, 1.42) & (1.05,
1.24)\\\hline
&  &  &  &  &  & \\
500/SFC & $\widehat{\theta}$ & 1.43 & 1.03 & 1.11 & 1.12 & 1.16\\
& CI & (1.41, 1.69) & (.909, 1.12) & (.973, 1.23) & (.999, 1.23) & (1.05,
1.27)\\\hline
&  &  &  &  &  & \\
400/SFC & $\widehat{\theta}$ & 1.55 & 1.260 & 1.40 & 1.46 & 1.31\\
& CI & (1.45, 1.77) & (1.11, 1.38) & (1.25, 1.55) & (1.36, 1.56) & (1.18,
1.43)\\\hline
&  &  &  &  &  & \\
300/SFC & $\widehat{\theta}$ & 1.62 & 1.36 & 1.62 & 1.68 & 1.79\\
& CI & (1.34, 1.81) & (1.140, 1.53) & (1.46, 1.78) & (1.52, 1.84) & (1.64,
1.93)\\\hline
&  &  &  &  &  & \\
200/SFC & $\widehat{\beta}$ & 1.58 & 1.38 & 1.38 & 1.72 & 1.84\\
& CI & (1.08, 1.66) & (1.19, 1.54) & (1.19, 1.54) & (1.54, 1.91) & (1.71,
1.98)\\\hline
&  &  &  &  &  & \\
150/SFC & $\widehat{\theta}$ & 1.38 & .751 & .424 & .924 & 1.04\\
& CI & (1.08, 1.66) & (.369, 1.10) & (.133, .679) & (.593, 1.19 & (.830,
1.23)\\\hline
&  &  &  &  &  & \\
100/SFC & $\widehat{\theta}$ & -.285 & -1.33 & -1.10 & -1.33 & -1.71\\
& CI & (-.968, .300) & (-2.92, -.138) & (-1.66, -.633) & (-3.69, -.048) &
(-2.82, -.849)\\\hline
&  &  &  &  &  & \\
70/SFC & $\widehat{\theta}$ & -2.18 & -2.62 & -3.10 & -2.44 & -2.74\\
& CI & (-3.28, .-1.10) & (-3.79, .-1.61) & (-4.26, .-2.09) & (-3.83, .-1.32) &
(-3.93, .-1.72)\\\hline
&  &  &  &  &  & \\
50/SFC & $\widehat{\theta}$ & -3.29 & -2.75 & -3.08 & -2.58 & -2.27\\
& CI & (-4.65, -2.20) & (-3.85, -1.82) & (-4.03, -2.31) & (-4.61, -1.76) &
(-2.85, -1.17)\\\hline
&  &  &  &  &  & \\
30/SFC & $\widehat{\theta}$ & -2.87 & -2.68 & -3.61 & -2.69 & -2.92\\
& CI & (-3.78, -2.18) & (-3.60, -2.00) & (-4.59, -2.86) & (-3.71, -1.91) &
(-3.61, -2.37)\\\hline
&  &  &  &  &  & \\
20/SFC & $\widehat{\beta}$ & -2.63 & -2.79 & NA & -2.74 & -3.93\\
& CI & (-3.52, -1.95) & (-3.70, -2.10) &  & (-3.61, -2.08) & (-4.81,
-3.26)\\\hline\hline
\multicolumn{7}{l}{Note: 95\% fixed-$b$ confidence intervals in brackets using
Daniell $k(x)$ function with Andrews}\\
\multicolumn{7}{l}{(1991) data dependent bandwidth. Data dependent bandwidth
sample size ratios}\\
\multicolumn{7}{l}{ranged from .018 to .135.}\end{tabular}
}\newpage\scalebox{0.9}{
\begin{tabular}
[c]{lcccccc}\multicolumn{7}{c}{Table 5: Differences in IV Estimated Trend Slope Ratios and
Product Differences by hPa Levels}\\
\multicolumn{7}{c}{Annual Data 1958-2024, $T=67$}\\\hline\hline
& \multicolumn{2}{c}{850/SFC} & \multicolumn{2}{c}{700/SFC} &
\multicolumn{2}{c}{500/SFC}\\\hline
& $\widehat{\Delta}_{\theta}$ & $g_{\widehat{\beta}}$ & $\widehat{\Delta
}_{\theta}$ & $g_{\widehat{\beta}}$ & $\widehat{\Delta}_{\theta}$ &
$g_{\widehat{\beta}}$\\\hline
RICH,RAOB & .255$^{\ast}$ & .517$^{\ast}$ & .308$^{\ast}$ & .624$^{\ast}$ &
.405$^{\ast}$ & .821$^{\ast}$\\
& (.177, .333) & (.342, .692) & (.254, .362) & (.480, .769) & (.280, .531) &
(.527, 1.12)\\\hline
RICH,RATP & -.026 & -.052 & .433$^{\ast}$ & .877$^{\ast}$ & .322$^{\ast}$ &
.653$^{\ast}$\\
& (-.179, .127) & (-.406, .301) & (.341, .525) & (596., 1.16) & (.210, .434) &
(.411, .895)\\\hline
RICH, ERA5 & -.052 & -.097 & .090 & .169 & .317$^{\ast}$ & .593$^{\ast}$\\
& (-.155, .051) & (-.282, .087) & (-.032, .213) & (-.153, .490) & (.184,
.450) & (.160, 1.03)\\\hline
RICH, JRA3Q & -.127$^{\ast}$ & -.261$^{\ast}$ & .226$^{\ast}$ & .464$^{\ast}$
& .273$^{\ast}$ & .561$^{\ast}$\\
& (-.235, -.019) & (-.491, -.032) & (.111, .340) & (.173, .754) & (.143,
.403) & (.221, .900)\\\hline
RAOB, RATP & -.281$^{\ast}$ & -.570$^{\ast}$ & .125$^{\ast}$ & .254$^{\ast}$ &
-.083$^{\ast}$ & -.169$^{\ast}$\\
& (-.418, -.144) & (-.908, -.232) & (.045, .205) & (022., .485) & (-.163,
-.004) & (-.323, -.015)\\\hline
RAOB,ERA5 & -.307$^{\ast}$ & -.576$^{\ast}$ & -.218$^{\ast}$ & -.408$^{\ast}$
& -.089 & -.166\\
& (-.438, -.176) & (-.576, -.792) & (-.345, -.090) & (-.642, -.175) & (-.185,
.008) & (-.381, .049)\\\hline
RAOB, JRA3Q & -.382$^{\ast}$ & -.787$^{\ast}$ & -.082 & -.170 & -.133$^{\ast}$
& -.273$^{\ast}$\\
& (-.484, -.281) & (-1.01, -.562) & (-.180, .019) & (-.407, .067) & (-.232,
-.033) & (-.524, -.022)\\\hline
RATP, ERA5 & -.026 & -.049 & -.343$^{\ast}$ & -.643$^{\ast}$ & -.006 & -.010\\
& (-.206, .153) & (-.423, .324) & (-.469, -.217) & (-.908, -.377) & (-.097,
.086) & (-.313, .292)\\\hline
RATP, JRA3Q & -.101 & -.209 & -.207$^{\ast}$ & -.427$^{\ast}$ & -.049 &
-.102\\
& (-.227, .024) & (-.453, .035) & (-.326, -.088) & (-.724, -.130) & (-.148,
.049) & (-.433, .229)\\\hline
ERA5, JRA3Q & -.075 & -.143 & .135$^{\ast}$ & .258 & -.044 & -.084\\
& (-.201, .051) & (-1.15, .866) & (.028, .243) & (-.036, .551) & (-.110,
.022) & (-.383, .216)\\\hline\hline
&  &  &  &  &  & \\
& \multicolumn{2}{c}{400/SFC} & \multicolumn{2}{c}{300/SFC} &
\multicolumn{2}{c}{200/SFC}\\\hline
& $\widehat{\Delta}_{\theta}$ & $g_{\widehat{\beta}}$ & $\widehat{\Delta
}_{\theta}$ & $g_{\widehat{\beta}}$ & $\widehat{\Delta}_{\theta}$ &
$g_{\widehat{\beta}}$\\\hline
RICH,RAOB & .291$^{\ast}$ & .589$^{\ast}$ & .265$^{\ast}$ & .536$^{\ast}$ &
.098$^{\ast}$ & .401\\
& (.176, .405) & (.315, .963) & (.146, .383) & (.265, .807) & (.037, .359) &
(-.195, .997)\\\hline
RICH,RATP & .145$^{\ast}$ & .295 & -.003 & -.006 & .464$^{\ast}$ & .940\\
& (.039, .252) & (-.006, .596) & (-.126, .120) & (-.343, .330) & (.205,
.723) & (-.352, 2.23)\\\hline
RICH, ERA5 & .087 & .163 & -.058 & -.108 & -.139 & -.259\\
& (-.036, .210) & (-.196, .522) & (-.247, .132) & (-.716, .499) & (-.437,
.160) & (-1.39, .871)\\\hline
RICH, JRA3Q & .239$^{\ast}$ & .491$^{\ast}$ & -.172$^{\ast}$ & -.353 &
-.262$^{\ast}$ & -.538$^{\ast}$\\
& (.095, .382) & (.068, .913) & (-.335, -.008) & (-.826, .119) & (-.501,
-.023) & (-1.37, .294)\\\hline
RAOB, RATP & -.145$^{\ast}$ & -.295$^{\ast}$ & -.268$^{\ast}$ & -.544$^{\ast}$
& .266$^{\ast}$ & .540$^{\ast}$\\
& (-.251, -.039) & (-.493, -.097) & (-.404, -.131) & (-.792, -.295) & (.099,
.433) & (-.020, 1.10)\\\hline
RAOB,ERA5 & -.203$^{\ast}$ & -.381$^{\ast}$ & -.322$^{\ast}$ & -.605$^{\ast}$
& -.336$^{\ast}$ & -.631$^{\ast}$\\
& (-.314, -.093) & (-.572, -.191) & (-.485, -.159) & (-.890, -.319) & (-.561,
-.112) & (-1.05, -.217)\\\hline
RAOB, JRA3Q & -.052 & -.107 & -.436$^{\ast}$ & -.899$^{\ast}$ & -.460$^{\ast}$
& -.947$^{\ast}$\\
& (-.154, .050) & (-.381, .167) & (-.583, -.290) & (-1.18, -.618) & (-.632,
-.288) & (-1.31, -.581)\\\hline
RATP, ERA5 & -.058 & -.109 & -.055 & -.103 & -.602$^{\ast}$ & -1.13$^{\ast}
$\\
& (-.156, .039) & (-.375, .157) & (-.186, .077) & (-.470, .264) & (-.791,
-.414) & (-1.52, -.739)\\\hline
RAPT, JRA3Q & .093 & .192 & -.169$^{\ast}$ & -.348 & -.726$^{\ast}$ &
-1.50$^{\ast}$\\
& (-.021,.207 ) & (-.148, .533) & (-.297, -.040) & (-.715, .019) & (-.896,
-.556) & (-2.09, -.896)\\\hline
ERA5, JRA3Q & .152$^{\ast}$ & .288$^{\ast}$ & -.114 & -.217 & -.123 & -.234\\
& (.091, .212) & (.183, .393) & (-.239, .011) & (-1.05, .618) & (-.317,
.071) & (-1.67, 1.21)\\\hline\hline
\multicolumn{7}{l}{Note: 95\% fixed-$b$ confidence intervals in brackets using
Daniell $k(x)$ function with Andrews (1991) data}\\
\multicolumn{7}{l}{dependent bandwidth. Data dependent bandwidth sample size
ratios ranged from .011 to .405. The values of $g_{\widehat{\beta}}$}\\
\multicolumn{7}{l}{and its confidence intervals are scaled by $10^{4}$ for
presentation purposes.}\end{tabular}
}\newpage\scalebox{0.9}{
\begin{tabular}
[c]{lcccccc}\multicolumn{7}{c}{Table 5 (continued): Differences in IV Estimated Trend
Slope Ratios and Product Differences by hPa Levels}\\
\multicolumn{7}{c}{Annual Data 1958-2024, $T=67$}\\\hline\hline
& \multicolumn{2}{c}{150/SFC} & \multicolumn{2}{c}{100/SFC} &
\multicolumn{2}{c}{70/SFC}\\\hline
& $\widehat{\Delta}_{\theta}$ & $g_{\widehat{\beta}}$ & $\widehat{\Delta
}_{\theta}$ & $g_{\widehat{\beta}}$ & $\widehat{\Delta}_{\theta}$ &
$g_{\widehat{\beta}}$\\\hline
RICH,RAOB & .624$^{\ast}$ & 1.26$^{\ast}$ & 1.05$^{\ast}$ & 2.12$^{\ast}$ &
.442 & .895\\
& (.352, .896) & (.228, 2.30) & (.305, 1.80) & (.493, 3.75) & (-.506, 1.39) &
(-.805, 2.60)\\\hline
RICH,RATP & .951$^{\ast}$ & 1.93$^{\ast}$ & .818$^{\ast}$ & 1.66$^{\ast}$ &
.920$^{\ast}$ & 1.86$^{\ast}$\\
& (.597, 1.31) & (.117, 3.74) & (.292, 1.35) & (.095, 3.22) & (.068, 1.77) &
(.267, 3.46)\\\hline
RICH, ERA5 & .451$^{\ast}$ & .843$^{\ast}$ & 1.05$^{\ast}$ & 1.96$^{\ast} $ &
.260 & .486\\
& (.185, .716) & (.139, 1.55) & (.095, 2.00) & (1.06, 2.86) & (-.830, 1.35) &
(-1.15, 2.13)\\\hline
RICH, JRA3Q & .331$^{\ast}$ & .681 & 1.42$^{\ast}$ & 2.92$^{\ast}$ & .559 &
1.149\\
& (.057, .605) & (-.298, 1.66) & (.868, 1.97) & (2.08, 3.76) & (-.500, 1.62) &
(-.627, 2.92)\\\hline
RAOB, RATP & .327 & .664 & -.229 & -.464 & .478$^{\ast}$ & .971$^{\ast}$\\
& (-.132,.786 ) & (-1.71, 3.04) & (-1.43, .970) & (-3.89, 2.96) & (.045,
.912) & (.084, 1.86)\\\hline
RAOB,ERA5 & -.173$^{\ast}$ & -.324$^{\ast}$ & .002 & .004 & -.182 & -.342\\
& (-.332, -.014) & (-.597, -.052) & (-.309, .314) & (-.448, .455) & (-.510,
.145) & (-.762, .079)\\\hline
RAOB, JRA3Q & -.193 & -.603 & .373 & .769 & .117 & .241\\
& (-.716, .131) & (-2.30, 1.09) & (-.590, 1.34) & (-1.94, 3.48) & (-.182,
.416) & (-.152, .634)\\\hline
RATP, ERA5 & -.500$^{\ast}$ & -.938 & .231 & .433 & -.661$^{\ast}$ & -1.24\\
& (-.875, -.125) & (-2.60, 0.721) & (-1.21, 1.68) & (-2.69, 3.55) & (-1.28,
-.043) & (-2.83, .357)\\\hline
RATP, JRA3Q & -.620$^{\ast}$ & -1.28$^{\ast}$ & .602 & 1.24 & -.361 & -.744\\
& (-.838, -.401) & (-2.13, -.424) & (-.128, 1.33) & (-.372, 2.85) & (-.915,
.192) & (-2.03, .538)\\\hline
ERA5, JRA3Q & -.120 & -.228 & .371 & .706 & .299 & .569\\
& (-.464, .225) & (-2.51., 2.05) & (-.750, 1.49) & (-6.23, 7.64) & (-.048,
.647) & (-.139, 1.28)\\\hline\hline
&  &  &  &  &  & \\
& \multicolumn{2}{c}{50/SFC} & \multicolumn{2}{c}{30/SFC} &
\multicolumn{2}{c}{20/SFC}\\\hline
& $\widehat{\Delta}_{\theta}$ & $g_{\widehat{\beta}}$ & $\widehat{\Delta
}_{\theta}$ & $g_{\widehat{\beta}}$ & $\widehat{\Delta}_{\theta}$ &
$g_{\widehat{\beta}}$\\\hline
RICH,RAOB & -.549$^{\ast}$ & -1.11$^{\ast}$ & -.196 & -.396 & .154 & .313\\
& (-.904, -.193) & (-1.73, -.494) & (-.392, .001) & (-1.84, 1.05) & (-.046,
.355) & (-1.88, 2.51)\\\hline
RICH,RATP & -.212 & -.431 & .733$^{\ast}$ & 1.49$^{\ast}$ & NA & NA\\
& (-.719, .294) & (-1.30, .434) & (.311, 1.16) & (.199, 2.77) &  & \\\hline
RICH, ERA5 & -.716$^{\ast}$ & -1.34$^{\ast}$ & -.189 & -.353 & .107 & .200\\
& (-1.15, -.286) & (-1.94, -.743) & (-.488, .111) & (-.943, .237) & (-.299,
.513) & (-.607, 1.01)\\\hline
RICH, JRA3Q & -1.02$^{\ast}$ & -2.10$^{\ast}$ & .049 & .101 & 1.298$^{\ast}$ &
2.67$^{\ast}$\\
& (-1.79, -.249) & (-2.97, -1.22) & (-.474, .572) & (-1.56, 1.76.) & (.630,
1.97) & (.676, 4.66)\\\hline
RAOB, RATP & .336 & .683 & .929$^{\ast}$ & 1.89$^{\ast}$ & NA & NA\\
& (-.023, .695) & (-.138, 1.50) & (.506, 1.35) & (.478, 3.30) &  & \\\hline
RAOB,ERA5 & -.167 & -.313 & .007 & .013 & -.048 & -.089\\
& (-.448, .114) & (-.732, .106) & (-.234, .248) & (-.430, .456) & (-.397,
.302) & (-.745, .566)\\\hline
RAOB, JRA3Q & -.472 & -.972$^{\ast}$ & .245 & .504 & 1.14$^{\ast}$ &
2.36$^{\ast}$\\
& (-1.042, .098) & (-1.77, -.178) & (-.239, .729) & (-1.09, 2.10) & (560,
1.72) & (.556, 4.15)\\\hline
RATP, ERA5 & -.503$^{\ast}$ & -.943 & -.922$^{\ast}$ & -1.73$^{\ast}$ & NA &
NA\\
& (-.970, -.036) & (-2.35, .467) & (-1.27, -.572) & (-3.22, -.235) &  &
\\\hline
RATP, JRA3Q & -.808$^{\ast}$ & -1.66$^{\ast}$ & -.684$^{\ast}$ & -1.41 & NA &
NA\\
& (-1.23, -.384) & (-2.40, -.926) & (-1.16, -.213) & (-3.07, .250) &  &
\\\hline
ERA5, JRA3Q & -.305 & -.580 & .238 & .452 & 1.19$^{\ast}$ & 2.27$^{\ast}$\\
& (-.874, .265) & (-2.06, .905) & (-.227, .703) & (-1.54, 2.44) & (.674,
1.71) & (.291, 4.24)\\\hline\hline
\multicolumn{7}{l}{Note: 95\% fixed-$b$ confidence intervals in brackets using
Daniell $k(x)$ function with Andrews (1991) data}\\
\multicolumn{7}{l}{dependent bandwidth. Data dependent bandwidth sample size
ratios ranged from .011 to .405. The values of $g_{\widehat{\beta}}$}\\
\multicolumn{7}{l}{and its confidence intervals are scaled by $10^{4}$ for
presentation purposes.}\end{tabular}
}
\end{document}